\documentclass{article}

\usepackage{amssymb,amsthm,latexsym,amsmath,graphics,setspace}
\usepackage{epsfig,epsf,amsfonts}
\usepackage{color,multicol,colordvi,graphicx,multirow}
\usepackage{fancyhdr}
\numberwithin{equation}{section}

\newtheorem{proposition}{Proposition}

\newtheorem{theorem}{Theorem}

\newcommand{\paren}[1]{\left(#1\right)}
\newcommand{\D}[2]{\frac{d#1}{d#2}}
\newcommand{\PD}[2]{\frac{\partial#1}{\partial#2}}

\newcommand{\mb}[1]{\mathbf{#1}}

\newcommand{\mc}[1]{\mathcal{#1}}
\newcommand{\abs}[1]{\left\lvert #1 \right\rvert}

\newcommand{\wh}[1]{\widehat{#1}}
\newcommand{\wt}[1]{\widetilde{#1}}

\title{A Multidomain Model for Ionic 
Electrodiffusion and Osmosis
with an Application to Cortical Spreading Depression}
\author{Yoichiro Mori\\
School of Mathematics\\
University of Minnesota}
\date{August 28, 2014}

\begin{document}
\maketitle

\begin{abstract}
Ionic electrodiffusion and osmotic water flow are central processes 
in many physiological systems. We formulate a system of 
partial differential equations that governs ion movement and water 
flow in biological tissue. A salient feature of this model is that 
it satisfies a free energy identity, ensuring the thermodynamic 
consistency of the model. A numerical scheme is developed for the 
model in one spatial dimension and is applied to a model of 
cortical spreading depression, a propagating breakdown of ionic 
and cell volume homeostasis in the brain.
\end{abstract}

\section{Introduction}

In this paper, we formulate a system of partial differential equations (PDE) 
that governs ionic electrodiffusion and osmotic water flow, 
to study tissue-level physiological phenomena.
To demonstrate the use of the model, we apply this to the study 
of cortical spreading depression, a pathological phenomenon of the 
brain that is linked to migraine aura and other diseases.

We now describe our modeling approach.
Biological tissue can often be seen as composed of multiple 
interpenetrating compartments. Cardiac tissue, for example, 
can be seen as composed of two interpenetrating compartments, 
the space that consists of interconnected cardiomyocytes 
and the extracellular space. 
The number of compartments may not be restricted to two. 
In the central nervous system, one may consider the neuronal, 
glial and extracellular compartments. 
In studying physiological phenomena at the tissue level,
it is often impractical to use models with exquisite 
cellular detail. If the spatial variations in the
biophysical variables of interest are slow compared to 
the cellular spatial scale, 
we may model the system instead as a homogenized continuum.
The first such model, the {\em bidomain model}, was introduced in 
\cite{eisenberg1970three,eisenberg1979electrical,tung1978bi},  
and its application to cardiac electrophysiology 
\cite{henriquez1992simulating} is probably 
the most important and successful example of this coarse-grained approach
in physiology. 
Let us use the cardiac bidomain model to further 
to illustrate this approach. The main variables of interest in 
cardiac electrophysiology are the intracellular and extracellular 
potentials, $\phi_i(\mb{x})$ and $\phi_e(\mb{x})$ where 
$\mb{x}$ is the spatial coordinate. From a microscopic standpoint, 
these values should only be defined within their respective compartments. 
At the coarse-grained level, however, we take the view that it is impossible 
to distinguish whether a given spatial point is inside the cell 
or outside the cell. The intracellular and extracellular potentials 
are now defined everywhere and cardiac tissue is thus seen 
as an biphasic continuum.  
In this paper, we shall call such models {\em multidomain models}
to emphasize the fact that the formalism is not restricted to 
just two interpenetrating phases.
We note that such coarse-grained models are also widely used in the 
material sciences to describe, for example, 
multiphase flow \cite{drew2012theory}.

Our goal is to formulate a multidomain model that describes 
ionic electrodiffusion and osmosis. This can be seen as a generalization 
of the cardiac bidomain model, which only treats electrical 
current flow. Ionic electrodiffusion and osmosis have been modeled 
to varying degrees of detail in different physiological systems.  
These include the kidney \cite{weinstein1994mathematical}, 
gastric mucosa \cite{lynch2011mathematical},
cerebral edema and hydrocephalus \cite{drapaca2012mechano}, 
cartilage \cite{gu1998mixture,gu1999transport}, 
and the lens \cite{malcolm2007computational} and cornea \cite{leung2011oxygen} of the eye. 
Here, we develop a time-dependent PDE model that fully incorporates
both ionic electrodiffusion and osmotic water flow in multiphasic tissue.
Ion balance is governed by the Nernst-Planck electrodiffusion equations
with source terms describing transmembrane ion flux. For water balance, 
we have the usual continuity equations with source terms describing 
transmembrane water flow. 
An important feature that distinguishes our model from previous models
is that it satisfies a free energy identity, which ensures that 
electrodiffusive and osmotic effects are treated
in a thermodynamically consistent fashion. 
The use of free energy identities as a guiding principle in formulating 
equations originates in the work of Onsager \cite{onsager1931reciprocal}, and this 
approach has been widely adopted in soft condensed matter physics 
\cite{doi1988theory,doi2011onsager,hyon2010energetic,eisenberg2010energy}.
The present work is closely related to our recent work in 
\cite{mori2011cvctrl,mori_liu_eis,mori2012gels,chen2014analysis},
wherein the free energy identity played an essential role in ionic 
electrodiffusion problems arising in physiology and the material sciences.
One practical benefit of the physically consistent formulation
of our model is that it treats fast cable (or electrotonic/electrical current) 
effects and the much slower effects mediated by ion concentration gradients
in a single unified framework.
This is significant especially in the context of ion homeostasis in the brain,
in which these fast and slow effects are both important and tightly coupled.

To demonstrate the use of the model (and to test our computational 
scheme), we have included a preliminary 
modeling study of cortical spreading depression (SD). 
SD is a pathological phenomenon of the central nervous system, 
first reported 70 years ago \cite{leao1944spreading}. 
Neurons sustain a complete 
depolarization and loss of functions for seconds to minutes. 
A massive redistribution of ions takes place \cite{grafstein1956mechanism} 
resulting in extracellular potassium concentrations in excess of 
$50$mmol/l. Also seen is neuronal swelling and narrowing of 
the extracellular space. 
This breakdown in ionic and volume homeostasis spreads across gray matter at 
speeds of $2-7$mm/min. SD is the physiological 
substrate of migraine aura, and it is also related to other 
brain pathologies such as stroke, seizures and trauma \cite{dreier2011role}. 
Studying SD is important, not only because of its close 
relationship with important diseases but also because a good 
understanding of SD will lead to a better understanding of 
brain ionic homeostasis, and hence of the workings of the 
central nervous system. 
Despite intensive research efforts, basic questions 
about SD remain unanswered \cite{miura2007cortical,herreras2005electrical}. 
We refer the reader to 
\cite{somjen2004ions,martins2000perspectives,
somjen2001mechanisms,charles2009cortical,
dahlem2014migraines}
for reviews on SD. 

There have been many modeling studies on SD propagation 
\cite{grafstein1963neuronal,reshodko1975computer,tuckwell1978mathematical,
tuckwell1981simplified,nicholson1993volume,
reggia1996computational,revett1998spreading,shapiro2001osmotic,
almeida2004modeling,bennett2008quantitative,
dahlem2010two,yao2011continuum,chang2013mathematical}, most of 
which are of reaction-diffusion type.
The large excursions in ionic concentration necessitates incorporation 
of ionic {\em electro}diffusion and osmotic effects, and 
our model is well-suited for this application.
As a natural output of our model, we can compute the negative 
shift in the extracellular potential (negative DC shift), 
an important experimental signal of SD. To the best of our 
knowledge, this is the first successful computation of this 
quantity. We then examine the effect of gap junctional coupling 
and extracellular chloride concentration on SD propagation speed. 
In particular, we argue that gap junctional coupling is unlikely 
to play an important role in SD propagation \cite{shapiro2001osmotic}.

The paper is organized as follows. In Section \ref{modelform}
we formulate the model. In Section \ref{sectFE}, we discuss the 
free energy identity. This identity allows us to place 
thermodynamic restrictions on the constitutive laws for 
the transmembrane fluxes. In Section \ref{sectsimple}, we make 
the equations dimensionless and discuss model reduction when 
certain dimensionless quantities are taken to $0$. In particular, 
we clarify the relationship between our multidomain electrodiffusion 
model with the cardiac bidomain model. In Section \ref{sectnum}, 
we discuss the numerical discretization of our system. 
We devise a implicit numerical method that preserves ionic concentrations and 
satisfies a discrete free energy inequality. In Section \ref{sectSD}, 
we perform simulations of SD. Appendix \ref{appSD} describes some 
of the details of the SD model and simulation and Appendix \ref{appExtVol}
includes some remarks on the computation of the extracellular voltage.

\section{Model Formulation}\label{modelform}

We suppose that the tissue of interest occupies a smooth bounded region 
$\Omega\in \mathbb{R}^3$. As discussed in the Introduction, 
we view biological tissue as being a multiphasic continuum. 
Suppose the tissue is composed of 
$N$ interpenetrating compartments which we label by $k$.
We assume that $k=N$ corresponds to the extracellular space
and that all other compartments communicate with the 
extracellular space only. When we only consider the intracellular and 
extracellular spaces, $N=2$ and the $2$nd compartment will be 
the extracellular space. In the central nervous system, we may 
consider neuronal, glial and extracellular spaces and the extracellular 
space corresponding to the $3$rd compartment, and the other two 
compartments communicating with the extracellular compartment. 
To each point in space, we assign a volume fraction $\alpha_k$
for each compartment. By definition, we have:
\begin{equation}
\sum_{k=1}^{N} \alpha_k(\mb{x},t)=1.\label{incomp}
\end{equation}
Note that $\alpha_k$ is a function of space and time.

In the following we shall introduce several parameters 
that may be influenced by the microscopic geometric details 
of the tissue. Mechanical properties of cells and 
hydraulic conductivity are examples of such parameters.
We shall make the assumption that these parameters depend on the underlying 
microscopic geometry only through its influence on $\alpha_k$.

In order to describe the time evolution of $\alpha_k$, 
we introduce the water flow velocity field $\mb{u}_k$
defined for each compartment. The volume fraction 
$\alpha_k$ satisfies the following equation:
\begin{align}
\PD{\alpha_k}{t}+\nabla \cdot (\alpha_k \mb{u}_k)&=- \gamma_k w_k, \; 
k=1,\cdots, N-1\label{alphak}\\
\PD{\alpha_{N}}{t}+\nabla \cdot(\alpha_N\mb{u}_N) 
&=\sum_{k=1}^{N-1}\gamma_k w_k \label{alphaN}
\end{align}
The coefficient 
$\gamma_k$ represents the area of cell membrane between compartment 
$k$ and the extracellular space per unit volume of tissue, and has units of 
$1/\text{length}$. We assume that the membrane does not stretch appreciably,
and take $\gamma_k$ to be constant in time. 
Transmembrane water flow per unit area of membrane is given 
by $w_k$ where flux going from compartment $k$ into the extracellular 
space is taken positive. Transmembrane water flow $w_k$ is 
a function of the volume fractions $\alpha_k$ as well as 
the ionic concentrations, the compartmental pressures and possibly 
the compartmental voltages, biophysical variables  
to be introduced below. This constitutive relation for $w_k$ will 
be discussed further in Section \ref{sectFE}. 
Equation \eqref{alphak} and \eqref{alphaN}, together with \eqref{incomp}
yields:
\begin{equation}\label{incompderiv}
\nabla \cdot \paren{\sum_{k=1}^N \alpha_k\mb{u}_k}=0.
\end{equation}
This condition states that the volume-fraction weighted velocity is divergence free, 
and corresponds to the incompressiblity condition for simple fluids.

We now turn to the dynamics of ionic concentrations. Let $c_i^k$
be the ionic concentration of the $i$-th species of ion in 
compartment $k$. 
We shall mainly be concerned with 
the inorganic ions (Na$^+$, K$^+$, Cl$^-$ etc) that
play an important role in electrophysiology and are
major contributors to osmotic pressure. Among the ions 
we do not track explicitly are the organic ions, including soluble proteins
and sugars and constituents of the intracellular and extracellular matrix. 
For simplicity, we neglect diffusion and transmembrane movement of 
these ions, which we call the immobile ions.
As we shall see, the background ions will exert electrostatic effects 
and contribute to osmotic pressure. 
We shall keep track of $M$ species of mobile ion. For each 
ionic species $i=1,\cdots,M$, we have the following conservation 
equations in each compartment.
\begin{align}
\PD{(\alpha_k c_i^k)}{t}&=-\nabla \cdot \mb{f}_i^k
-\gamma_kg_i^k, \label{cik}\; k=1,\cdots,N-1,\\
\PD{(\alpha_N c_i^N)}{t}&=-\nabla \cdot \mb{f}_i^N
+\sum_{k=1}^{N-1}\gamma_kg_i^k,\label{ciN}\\
\mb{f}_i^k&=-D_i^k\paren{\nabla c_i^k+\frac{z_iFc_i^k}{RT}\nabla \phi_k}
+\alpha_k\mb{u}_kc_i^k, \; k=1,\cdots,N.\label{fik}
\end{align}
In these equations,
$F$ is the Faraday constant, 
$D_i^k$ is the diffusion coefficient, 
$z_i$ is the valence of the $i$-th species of ion, 
$RT$ is the ideal gas constant times absolute temperature,
and $\phi^k$ is the electrostatic potential of the $k$-th compartment.
The diffusion coefficient $D_i^k$ is in general a diffusion tensor 
that may be a function of $\alpha_k$.
The terms $g_i^k$ in \eqref{cik} and \eqref{ciN} are 
the transmembrane fluxes per unit membrane area for each species of ion. 
Biophysically, these are fluxes that flow through ion channels, 
transporters, or 
pumps that are located on the cell membrane. It is useful to split 
this transmembrane flux into two terms:
\begin{equation}\label{gjh}
g_i^k=j_i^k+h_i^k.
\end{equation}
The flux $j_i^k$ is the passive flux corresponding to ion channel and 
transporter fluxes. The flux $h_i^k$ is the active flux through 
ionic pumps. Both $j_i^k$ and $h_i^k$ are functions of the ionic 
concentrations, compartmental voltage, and possibly the volume 
fractions and the compartmental pressure. The compartmental pressure $p_k$
will be introduced shortly. 
Ion channel currents are often also controlled by channel gating, 
and in such cases, $j_i^k$ will also depend on gating variables.
The constitutive relations for $j_i^k$
and $h_i^k$ will be discussed further in Section \ref{sectFE}, where we 
give a precise definition of what is meant by a passive flux. 

To specify the electrostatic potential $\phi^k$, we have the following 
equations which we call the {\em charge capacitor relation}:
\begin{align}
\gamma_kC_{\rm m}^k\phi_{kN}&=z_0^kFa_k+\sum_{i=1}^Mz_iF\alpha_kc_i^k, \; \phi_{kN}=\phi_k-\phi_N,\;
k=1,\cdots,N-1,\label{cmk}\\
-\sum_{k=1}^{N-1}\gamma_kC_{\rm m}^k\phi_{kN}&=z_0^NFa_N+\sum_{i=1}^Mz_iF\alpha_Nc_i^N\label{cmN}
\end{align}
These equations state
that excess charge is stored on the membrane capacitor. The constant 
$C_{\rm m}^k$ is the membrane capacitance per unit area of membrane
separating the $k$-th and $N$-th compartment. The immobile charge
density is given by $z_0^kFa_k$ where $z_0^k$ and $a_k$
are the valence and amount of immobile solutes respectively. 
We assume that the $a_k$ are constant in time. Given the smallness 
of the capacitance, 
it is often an excellent approximation 
to use the following electroneutrality condition in place of \eqref{cmk} and \eqref{cmN}:
\begin{equation}
z_0^kFa_k+\sum_{i=1}^Mz_iF\alpha_kc_i^k=0, \; k=1,\cdots,N.\label{EN}
\end{equation}
We shall come back to this approximation when we discuss non-dimensionalization
in Section \ref{sectsimple}. The charge capacitor relation can, thus, also be 
considered a condition for near electroneutrality. 
Under the electroneutrality approximation, $\phi_k$ is determined so that 
the electroneutrality condition is satisfied.
A differential equation for $\phi_k$ may be obtained 
by taking the time derivative of \eqref{EN} with 
respect to $t$ and using \eqref{cik} and \eqref{ciN}. 
We shall discuss this further later on.

We also point out that the charge capacitor relation of \eqref{cmk} 
and \eqref{cmN}
plays the role of the Poisson equation in the Poisson-Nernst-Planck system, 
in that \eqref{cmk} and \eqref{cmN} determine the electrostatic potential.
The use of this relationship in pump-leak model is standard 
\cite{hoppensteadt2002modeling,KS}. 
Its use in a spatially extended context appears in 
\cite{Qian-Sej1,koch1999biophysics}. We also point 
to \cite{mori2008ephaptic,mori2009numerical} 
in which similar relations are used. The use of the 
the charge capacitor relation in place of the Poisson equation 
is warranted in part because the space charge layer (Debye length, 
typically on the order of nanometers) 
is very small compared even to the cellular length scale. Indeed, 
much of the interest in applications of the Poisson-Nernst-Planck system 
in biology concerns modeling of ion channels 
and other biomolecules \cite{nonner_progress_1999,wei2012variational}, 
a problem at much smaller length scales than the problem at hand. 

Let us turn to the equations for $\mb{u}_k$.
We introduce the compartmental pressure fields 
$p_k$. 
\begin{equation}
\zeta_k\mb{u}_k=-\nabla \wt{p}_k-\sum_{i=1}^Mz_iFc_i^k\nabla \phi_k, \;
\wt{p}_k=p_k-RT\frac{a_k}{\alpha_k}, \; k=1\cdots N.\label{darcy}
\end{equation}
Here, $\zeta_k$ is the hydraulic resistivity for the $k$-th compartment
and $a_k$ is the amount of immobile ions in the $k$-th compartment.
The above states that the flow is driven by electrostatic forces and 
the modified pressure $\wt{p}_k$. The modified pressure $\wt{p}_k$ has 
a mechanical contribution $p_k$ as well as a contribution from 
the immobile ions $a_k/\alpha_k$. The $a_k/\alpha_k$ term is 
known as the {\em oncotic pressure} 
in the physiology literature \cite{boron2008medical}.
The hydraulic resistivity $\zeta_k$ is in general 
a position dependent tensor, but 
we may, for simplicity, assume that $\zeta_k$ is a scalar.
For the extracellular space, a simple prescription may be to 
set $\zeta_k$ proportional to $\alpha_k$. In the case of 
the intracellular space, hydraulic resistivity in many tissues 
should be controlled by gap junctions connecting adjacent cells. 
In the absence of gap junctions, $(\zeta_k)^{-1}$ may be set to $0$.  

To determine the compartmental pressures $p_k$, we consider
force balance between compartment $k$ and the extracellular 
space. This leads to the following expression:
\begin{equation}
p_k-p_N=\tau_k(\alpha), \; k=1,\cdots,N-1\label{pk}
\end{equation}
where $\tau_k$ is the mechanical tension per unit area of the membrane 
separating 
compartment $k$ and the extracellular space. 
The membrane tension $\tau_k$ 
should be determined by the instantaneous microscopic configuration of the 
membrane. Given our assumption that the effects of microscopic geometry manifests
itself only through its influence on $\alpha$,  
$\tau_k$ must be given as a function of the volume fractions 
$\alpha=(\alpha_1,\cdots,\alpha_N)$.
A simple constitutive relation may be:
\begin{equation}
\tau_k=S_k(\alpha_k-\alpha_k^0)\label{taukconst}
\end{equation}
where $\alpha_k^0$ is the volume fraction at which the membrane 
has no mechanical tension and $S_k$ is a stiffness constant. 
We consider a class of constitutive relations that can be derived from 
some energy function $\mathcal{E}(\alpha_1,\cdots,\alpha_{N-1})$ in the following sense:
\begin{equation}\label{taukelas}
\tau_k(\alpha)=\PD{\mathcal{E}}{\alpha_k}.
\end{equation}
The simple constitutive relation \eqref{taukconst} clearly satisfies
condition \eqref{taukelas} with the choice:
\begin{equation}
\mathcal{E}=\frac{1}{2}\sum_{k=1}^{N-1}S_k(\alpha_k-\alpha_k^0)^2.
\end{equation}

We have only specified the constitutive relation for the 
difference $p_k-p_N$. The extracellular pressure $p_N$ is determined 
so that the incompressibility condition \eqref{incompderiv} is satisfied.
We may derive an equation for $p_N$ by multiplying \eqref{darcy} by $\alpha_k(\zeta_k)^{-1}$, 
taking the divergence and taking the summation in $k=1,\cdots,N$. 
We obtain:
\begin{equation}\label{pNeqn}
0=\nabla \cdot\paren{\sum_{k=1}^N\paren{\alpha_k\zeta_k^{-1}
\paren{\nabla \paren{\tau_k(\alpha)+p_N-\frac{RTa_k}{\alpha_k}}+\sum_{i=1}^N z_iFc_i^k\nabla\phi_k}}},
\end{equation}
where we set $\tau_N=0$ for notational convenience
and used \eqref{incompderiv} to obtain $0$ on the left hand side of the above.

Boundary conditions will 
strongly depend on the problem in question. In this paper 
we shall assume no flux 
boundary conditions at the boundary $\partial \Omega$:
\begin{equation}
\mb{u}_k\cdot \mb{n}=0, \;\; \mb{f}_i^k\cdot \mb{n}=0
\label{bc}
\end{equation}
where $\mb{n}$ is the outward unit normal on $\partial \Omega$.

In the above, our region $\Omega$ was a bounded region in $\mathbb{R}^3$.
It is also meaningful to consider the above equations in a one-dimensional 
or two-dimensional region. This corresponds to a problem in which the biophysical 
variables of interest are assumed to have no spatial dependence in two or 
one coordinate direction respectively. 
Most of the calculations to follow remain valid when $\Omega$ is a 1D or 2D 
region instead of a 3D region.
In Section \ref{sectnum}, we present a numerical simulation for a 1D version of the model.

\section{A Free Energy Identity}\label{sectFE}

We shall now state and prove a free energy identity for the above system of equations.
Before we state the energy identity, we define some useful quantities.
\begin{align}
\mu_i^k&=RT(\ln c_i^k+1)+z_iFc_i^k\phi_k, \label{muik}\\
\pi_{{\rm w}k}&=RT\paren{\frac{a_k}{\alpha_k}+\sum_{i=1}^M c_i^k}.\label{piwk} 
\end{align}
The quantity $\mu_i^k$ is the chemical potential of the $i$-th species of ion 
in the $k$-th compartment.
The quantity $\pi_{{\rm w}k}$ is the osmotic pressure. It is also useful to define the 
following water potential:
\begin{equation}\label{psik}
\psi_k=p_k-\pi_{{\rm w}k}.
\end{equation}
\begin{theorem}\label{FEthm}
Suppose $\alpha_k,\mb{u}_k,c_i^k,\phi_k$ and $p_k$ are smooth functions that satisfy 
\eqref{incomp}, \eqref{alphak}, \eqref{alphaN}, \eqref{cik}, \eqref{ciN}, \eqref{cmk}, \eqref{cmN},
\eqref{darcy}, \eqref{pk}, \eqref{taukelas} and \eqref{bc}. 
Then, the following identity holds. 
\begin{equation}\label{FE}
\begin{split}
\D{G}{t}&=-I_{\rm bulk}-I_{mem}, \\
G&=\int_\Omega \paren{\mathcal{E}+\sum_{k=1}^N\paren{RT\paren{a_k\ln \paren{\frac{a_k}{\alpha_k}}+\sum_{i=1}^M \alpha_kc_i^k\ln c_i^k}}
+\sum_{k=1}^{N-1}\frac{1}{2}\gamma_kC_{\rm m}\phi_{kN}^2}d\mb{x},\\
I_{\rm bulk}&=\int_\Omega \paren{\sum_{k=1}^N\paren{\alpha_k\zeta_k\abs{\mb{u}_k}^2+\sum_{i=1}^M\frac{D_i^kc_i^k}{RT}\abs{\nabla \mu_i^k}^2}}d\mb{x},\\
I_{\rm mem}&=\int_\Omega \paren{\sum_{k=1}^{N-1}\gamma_k \paren{\psi_{kN}w_k+\sum_{i=1}^M\mu_i^{kN}g_i^k}}d\mb{x},
\end{split}
\end{equation}
where $\psi_{kN}=\psi_k-\psi_N$ and $\mu_i^{kN}=\mu_i^k-\mu_i^N$.
\end{theorem}
In \eqref{FE}, the function $G$ should be interpreted as the free energy of the system, 
given as the sum of the elastic energy, the free energy from the ions and the 
electrical energy stored on the membrane capacitor. The change in $G$ is written 
as a sum of two parts, $-I_{\rm bulk}$, arising from biophysical 
processes within each compartment, and, $-I_{\rm mem}$, across the cell membranes.
\begin{proof}
Multiply both sides of \eqref{cik} by $\mu_i^k$ and integrate over $\Omega$.
The left hand side yields:
\begin{equation}\label{mucik1}
\begin{split}
\int_\Omega \mu_i^k \PD{(\alpha_kc_i^k)}{t}d\mb{x}
=\int_\Omega \paren{RT\paren{\PD{}{t}(\alpha_kc_i^k\ln c_i^k)+c_i^k\PD{\alpha_k}{t}}
+z_kF\phi_k\PD{(\alpha_kc_i^k)}{t}}d\mb{x}
\end{split}
\end{equation}
The left hand side for \eqref{cik} yields:
\begin{equation}\label{mucik2}
\begin{split}
&-\int_\Omega \mu_i^k(\nabla \cdot \mb{f}_i^k+\gamma_k g_i^k)d\mb{x}=\int_\Omega (\mb{f}_i^k\cdot \nabla \mu_i^k-\gamma_k \mu_i^kg_i^k)d\mb{x}\\
=&\int_\Omega \paren{-\frac{D_i^kc_i^k}{RT}\abs{\nabla \mu_i^k}^2+RT\alpha_k\mb{u}_k\cdot\nabla c_i^k+z_iF\alpha_kc_i^k\mb{u}_k\cdot\nabla \phi_k-\gamma_k \mu_i^kg_i^k}d\mb{x}\\
=&\int_\Omega \paren{-\frac{D_i^kc_i^k}{RT}\abs{\nabla \mu_i^k}^2-RTc_i^k\nabla \cdot (\alpha_k\mb{u}_k)
+z_iF\alpha_kc_i^k\mb{u}_k\nabla \phi_k-\gamma_k \mu_i^kg_i^k}d\mb{x}.
\end{split}
\end{equation}
In the above, we integrated by parts and used \eqref{bc} in the first equality,
used \eqref{fik} and \eqref{muik} in the second equality and integrated by parts and used \eqref{bc} 
in the last equality. Combining \eqref{mucik1} and \eqref{mucik2} and using \eqref{alphak}, we find:
\begin{equation}\label{mucik3}
\begin{split}
&\int_\Omega \paren{RT\PD{}{t}\paren{\alpha_kc_i^k\ln c_i^k}
+z_iF\phi_k\PD{(\alpha_kc_i^k)}{t}}d\mb{x}\\
=&\int_\Omega \paren{-\frac{D_i^k}{RT}\abs{\nabla \mu_i^k}^2+RTc_i^k\gamma_kw_k
+z_iF\alpha_kc_i^k\mb{u}_k\cdot\nabla \phi_k-\gamma_k \mu_i^kg_i^k}d\mb{x}
\end{split}
\end{equation}
We now take the summation in $i=1,\cdots,M$ on both sides of the above. 
Note that:
\begin{equation}\label{mucik4}
\sum_{i=1}^M z_iF\phi_k\PD{(\alpha_kc_i^k)}{t}=\gamma_kC_{\rm m}\phi_k\PD{\phi_{kN}}{t}.
\end{equation}
where we used \eqref{cmk}. Furthermore, we have:
\begin{equation}\label{mucik5}
\begin{split}
&\int_\Omega \paren{\sum_{i=1}^M z_iF\alpha_kc_i^k\mb{u}_k\cdot\nabla \phi_k}d\mb{x}
=-\int_\Omega\paren{\alpha_k\zeta_k\abs{\mb{u}_k}^2+\alpha_k\mb{u}_k\cdot\nabla\wt{p}_k}d\mb{x} \\
=&\int_\Omega\paren{-\alpha_k\zeta_k\abs{\mb{u}_k}^2+\nabla \cdot(\alpha_k\mb{u}_k)\wt{p}_k}d\mb{x}\\
=&\int_\Omega\paren{-\alpha_k\zeta_k\abs{\mb{u}_k}^2-\paren{\PD{\alpha_k}{t}+\gamma_kw_k}\wt{p}_k}d\mb{x}\\
=&\int_\Omega\paren{-\alpha_k\zeta_k\abs{\mb{u}_k}^2-p_k\PD{\alpha_k}{t}
-\PD{}{t}\paren{RTa_k\ln \paren{\frac{a_k}{\alpha_k}}}-\gamma_kw_k\wt{p}_k}d\mb{x}.
\end{split}
\end{equation}
where we used \eqref{darcy} in the first equality, integrated by parts in the second equality, 
used \eqref{alphak} in the third equality and the definition of $\wt{p}_k$ in \eqref{darcy}
in the last equality.
We may now use \eqref{mucik4} and \eqref{mucik5} with \eqref{mucik3} to find that
\begin{equation}\label{mucik}
\begin{split}
&\int_\Omega \paren{RT\PD{}{t}\paren{a_k\ln \paren{\frac{a_k}{\alpha_k}}+\sum_{i=1}^M \alpha_kc_i^k\ln c_i^k}
+\gamma_kC_{\rm m}\phi_k\PD{\phi_{kN}}{t}}d\mb{x}\\
=&-\int_\Omega \paren{\alpha_k\zeta_k\abs{\mb{u}_k}^2+\sum_{i=1}^M\frac{D_i^kc_i^k}{RT}\abs{\nabla \mu_i^k}^2}d\mb{x}\\
&+\int_\Omega\paren{-p_k\PD{\alpha_k}{t}+\gamma_k\paren{\psi_kw_k+\sum_{i=1}^M\mu_i^k g_i^k}}d\mb{x}.
\end{split}
\end{equation}
where we used \eqref{piwk}, \eqref{psik} and the definition of $\wt{p}_k$ in \eqref{darcy}.
The above equation is valid for $k=1,\cdots, N-1$. For $k=N$, we may derive a relation similar to \eqref{mucik}
by multiplying \eqref{ciN} with $\mu_i^N$ and taking the sum in $i=1,\cdots M$. This yields:
\begin{equation}\label{muciN}
\begin{split}
&\int_\Omega \paren{RT\PD{}{t}\paren{a_N\ln \paren{\frac{a_N}{\alpha_N}}+\sum_{i=1}^M \alpha_Nc_i^N\ln c_i^N}
-\sum_{k=1}^{N-1}\gamma_kC_{\rm m}\phi_N\PD{\phi_{kN}}{t}}d\mb{x}\\
=&-\int_\Omega \paren{\alpha_N\zeta_N\abs{\mb{u}_N}^2+\sum_{i=1}^M\frac{D_i^Nc_i^N}{RT}\abs{\nabla \mu_i^N}^2}d\mb{x}\\
&+\int_\Omega\paren{-p_N\PD{\alpha_N}{t}-\sum_{k=1}^{N-1}\gamma_k\paren{\psi_Nw_k+\sum_{i=1}^M\mu_i^N g_i^k}}d\mb{x}.
\end{split}
\end{equation}
Take the summation of both sides of \eqref{mucik} in $k=1,\cdots,N-1$ and add this to both sides 
of \eqref{muciN}. This computation yields \eqref{FE} by noting that:
\begin{equation}\label{enderiv}
\sum_{k=1}^Np_k\PD{\alpha_k}{t}=\sum_{k=1}^{N-1}(p_k-p_N)\PD{\alpha_k}{t}=\sum_{k=1}^{N-1}\tau_k\PD{\alpha_k}{t}=\PD{\mathcal{E}}{t},
\end{equation}
where we used \eqref{incomp} in the first equality, \eqref{pk} in the second equality and \eqref{taukelas} in the third equality.
\end{proof}

In the above energy identity \eqref{FE}, $I_{\rm bulk}$ is non-negative, and therefore, 
leads to dissipation in free energy. If $I_{\rm mem}$ is also non-negative, 
then the free energy $G$ will be non-increasing. 
Substitute \eqref{gjh} into the expression for $I_{\rm mem}$ in \eqref{FE}.
\begin{equation}
\begin{split}
I_{\rm mem}&=I^{\rm passive}_{\rm mem}+I^{\rm active}_{\rm mem},\\
I^{\rm passive}_{\rm mem}&=
\sum_{k=1}^{N-1}\int_\Omega \gamma_k \paren{\psi_{kN}w_k+\sum_{i=1}^M\mu_i^{kN}j_i^k}d\mb{x},\\
I^{\rm active}_{\rm mem}&=
\sum_{k=1}^{N-1}\int_\Omega \gamma_k \paren{\sum_{i=1}^M\mu_i^{kN}h_i^k}d\mb{x}.
\end{split}
\end{equation}
Given the above expression, we require that 
the water flux $w_k$ and the passive (or dissipative) ionic flux $j_i^k$
satisfy the following inequality:
\begin{equation}\label{dissip}
\psi_{kN}w_k+\sum_{i=1}^n\mu_i^{kN}j_i^k\geq 0, \; k=1,\cdots,N-1.
\end{equation}
With inequality \eqref{dissip}, $I^{\rm passive}_{\rm mem}$ is always positive 
and leads to free energy dissipation whereas $I^{\rm active}_{\rm mem}$ may 
lead to either free energy increase or decrease.
We have assumed here that the water flux $w_k$ is wholly passive,
since there seems to be little experimental evidence of a molecular water pump.
There is no mathematical difficulty in introducing an active water flux however; 
all that needs to be done is to split the transmembrane water flux into 
an active and passive component as in \eqref{gjh}. 

From a biophysical standpoint, 
a slightly better definition of dissipativity may be given as follows.
Passive ionic flux is carried by different types of ion channels and transporters. Water flux is carried by water channels (aquaporins) or directly through the lipid bilayer membrane.
Suppose that there are $m=1,\cdots,N_{\rm c}$ types of channels or transporters
(we may also add a label to the lipid bilayer membrane itself, 
if water flux through it is non-negligible).
Then, the transmembrane water flux and ion channel flux may be written as
\begin{equation}\label{waterionmemflux}
w_k=\sum_{m=1}^{N_{\rm c}}w_{km}, \; j_i^k=\sum_{m=1}^{N_{\rm c}} j_{im}^k,
\end{equation}
where $w_{km}$ and $j^k_{im}$ are the transmembrane water flux and ion flux
for the $i$-th species of ion 
carried by channel/transporter type $m$ residing in cell membrane $k$.
For each $m$, we require that
\begin{equation}\label{dissipmol}
\psi_{kN}w_{km}+\sum_{i=1}^n\mu_i^{kN}j_{im}^k\geq 0, \; k=1,\cdots,N-1.
\end{equation}
If \eqref{dissipmol} is satisfied, \eqref{dissip} is clearly satisfied.
Suppose that a particular channel type $m$ is permeable only to a single 
species of ion $i=i'$ and is not permeable to water. Then, 
$j_{im}^k=0$ for $i\neq i'$ and $w_{km}=0$, and therefore, 
there is 
only one term in the left hand side of \eqref{dissipmol}:
\begin{equation}\label{dissipsingle}
\mu_i^{kN}j_{i'm}^k\geq 0.
\end{equation}
This implies that $j_{i'm}^k$ must have the same sign as $\mu_{i'}^{kN}$. 
In physico-chemical terms, this states that the ionic flux flows
from where the chemical potential is high to low. It is in this 
sense that $j_{i'm}^k$ is a passive flux. 

Typical constitutive relations
for ion channel flux has the form:
\begin{equation}\label{jimk}
j_{im}^k(\mb{x},\mb{s}_m^k,\mb{c}^k,\mb{c}^N,\phi^{kN})=g_{im}^k(\mb{x},\mb{s}_m^k)J_{im}(\mb{c}^k,\mb{c}^N,\phi_{kN}),
\end{equation}
where $\mb{c}^k=(c_1^k,\cdots,c_M^k), \mb{c}^N=(c_1^N,\cdots,c_M^N)$
and $\mb{s}_m^k=(s_{m1}^k,\cdots,s_{mG}^k)$ are the gating variables which specify 
the proportion of ion channels that are open.
The function $g_k(\mb{x},\mb{s}_m^k)$ denotes the density of open 
channels in cell membrane $k$ at location $\mb{x}$. 
The function 
$J_{im}$, when converted to units of electrical current rather than flux,
is known as the instantaneous current-voltage relationship.
The simplest choice may be the linear current voltage relation
\begin{equation}
J_{im}^\text{lin}=G_{im}\mu_i^{kN}=
G_{im}\paren{RT\ln\paren{\frac{c_i^k}{c_i^N}}+z_iF\phi_{kN}},\label{HH}
\end{equation}
where $G_{im}>0$ and $G_{im}(z_iF)^2$ is the conductance.  
The following Goldmain-Hodgkin-Katz relation is also used very often.
\begin{equation}
\begin{split}
J_{im}^\text{GHK}&=P_{im}J_{\rm GHK}(z_i,c_i^k,c_i^N,\phi_{kN}),\\
J_{\rm GHK}&=z_i\phi'\paren{
\frac{c_i^k\exp(z_i\phi')-c_i^N}{\exp(z_i\phi')-1}}, 
\; \phi'=\frac{F\phi_{kN}}{RT},\label{GHK}
\end{split}
\end{equation}
where $P_{im}>0$ is known as the permeability.
Many ion channels are selectively permeable to one species of ion $i=i'$.
Such a channel type $m$ may be modeled so that
$\eta_{im}$ (or $\tilde{\eta}_{im}$) is non-zero only for 
$i=i'$ and $w_{km}=0$. It is easily seen that both \eqref{HH}
and \eqref{GHK} satisfy condition \eqref{dissipsingle}.

The gating variables
$\mb{s}_m^k=(s^k_{m1},\cdots,s^k_{mG})$ that appear in \eqref{jimk} 
satisfy an ODE of the form:
\begin{equation}\label{gateeqn}
\PD{s_{mg}^k}{t}=Q_{mg}(s_{mg}^k,\mb{c}^k,\mb{c}^N,\phi_{kN}).
\end{equation}
Typically, $Q_{mg}$ is a linear function of $s_{mg}^k$ and depends only 
on $\phi_{kN}$.
Examples of \eqref{HH}, \eqref{GHK} and  
are used in the computational examples 
discussed in Section \ref{sectSD}.


Some transporters couple the flow of two or more different ionic species
in the sense that the chemical potential difference of ion $i$ may influence 
the flow of ion $i', i\neq i'$. Flux through such a passive transporter
will not in general satisfy \eqref{dissipsingle} but must still satisfy 
the more general relation \eqref{dissipmol}. 
Examples of such transporter models can be found, for example, 
in \cite{weinstein1994ammonia}.

There are no thermodynamic restrictions on the constitutive relation for the active 
flux $h_i^k$. The flux $h_i^k$ may consist of fluxes carried by different ionic 
pumps, and thus, may have the form:
\begin{equation}\label{hik}
h_i^k=\sum_{m=1}^{N_{\rm p}} h_{im}^k(\mb{x},\mb{c}^k,\mb{c}^N,\phi_{kN}).
\end{equation} 

Let us now turn to the constitutive relation for the passive water flux 
$w_{km}$. If the water flow is not influenced by the chemical potential 
difference of other ions, \eqref{dissipmol} implies that $w_{km}$
must satisfy:
\begin{equation}
\psi_{kN}w_{km}\geq 0.
\end{equation}
This means that water flows from where the water potential $\psi$
is high to low. The water potential, defined in \eqref{psik}, is 
given as the difference between the mechanical and osmotic pressures. 
We thus arrive at the familiar statement that water flow is 
driven by a competition of mechanical and osmotic pressures.
A simple prescription for $w_{km}$ is:
\begin{equation}\label{linwaterflow}
w_{km}(\mb{x},\mb{c}^k,\mb{c}^N,\alpha_k,\alpha_N)=\eta^{\rm w}_{km}(\mb{x})\psi_{kN},
\end{equation}
where $\eta^{\rm w}_{km}$ is the hydraulic permeability.
If water flow is influenced by the chemical potential difference of 
ions, the more general \eqref{dissipmol} is satisfied.
If the chemical potential of ions influence water flow, 
Onsager reciprocity implies 
that water potential must have an influence 
ion flux \cite{katzir1965nonequilibrium}.
The effect of water flow on ion flux is known as solvent drag 
\cite{boron2008medical}.  

\section{Simplifications}\label{sectsimple}

The model we just described incorporates 
effects of electrodiffusion, osmosis, volume changes 
and water flow in a three dimensional setting. 
However, we do not expect all of these effects to 
be important in all physiological systems of interest.
It is thus of interest to see how the model simplifies when 
a subset of these effects are deemed negligible.  

We first make our system dimensionless. 
We introduce the following rescaling. 
\begin{equation}
x=L\wh{x}, \; t=\tau_D\wh{t}=\frac{L^2}{D_0}\wh{t}, \; c_i^k=c_0\wh{c}_i^k, \;
\phi=\frac{RT}{F}\wh{\phi},\;
\mb{u}_k=\frac{c_0RT}{\zeta_0}\wh{\mb{u}}_k,
\end{equation}
where $\wh{\cdot}$ denotes the dimensionless variables. 
In the above, $L$ is the characteristic domain size, 
$D_0, c_0$ are the typical magnitude of the diffusion coefficient
and concentrations respectively  
and  $\zeta_0$ is the representative magnitude 
of the hydraulic resistivity (the coefficients $\zeta_k$ in \eqref{darcy}).
With the above dimensionless variables, we may rewrite equations 
\eqref{alphak}, \eqref{alphaN}, \eqref{cik}, \eqref{ciN} as follows.
\begin{align}\label{akdless}
\PD{\alpha_k}{\wh{t}}+{\rm Pe}\wh{\nabla}\cdot(\alpha_k\wh{\mb{u}}_k)
&=-\wh{w}_k\\\label{aNdless}
\PD{\alpha_N}{\wh{t}}+{\rm Pe}\wh{\nabla}\cdot(\alpha_N\wh{\mb{u}}_N)
&=\sum_{k=1}^{N-1}\wh{w}_k\\
\label{cikdless}
\PD{(\alpha_k\wh{c}_i^k)}{\wh{t}}
+{\rm Pe}\wh{\nabla}\cdot(\alpha_k\wh{\mb{u}}_k\wh{c}_i^k)
&=\wh{\nabla}\cdot\paren{\wh{D}_i^k\paren{\wh{\nabla}\wh{c}_i^k+z_i\wh{c}_i^k\wh{\nabla}\wh{\phi}_k}}
-\wh{g}_i^k\\
\label{ciNdless}
\PD{(\alpha_N\wh{c}_i^N)}{\wh{t}}
+{\rm Pe}\wh{\nabla}\cdot(\alpha_N\wh{\mb{u}}_k\wh{c}_i^N)
&=\wh{\nabla}\cdot\paren{\wh{D}_i^N\paren{\wh{\nabla}\wh{c}_i^N+z_i\wh{c}_i^N\wh{\nabla}\wh{\phi}_N}}
+\sum_{k=1}^{N-1}\wh{g}_i^k
\end{align}
where 
\begin{equation}
D_k=D_0\wh{D}_k, \; 
\gamma_kw_k=\frac{1}{\tau_D}\wh{w}_k, \; 
\gamma_k\wh{g}_i^k=\frac{c_0}{\tau_D}\wh{g}_i^k, 
\; {\rm Pe}=\frac{c_0RT/\zeta_0}{L/\tau_D}.
\end{equation}
The dimensionless number ${\rm Pe}$ is the P\'eclet number in which 
the representative fluid velocity is taken to be $c_0RT/\zeta_0$.
To make \eqref{cmk}, \eqref{cmN} dimensionless,
we introduce the following dimensionless variables.
\begin{equation}
a_k=c_0\wh{a}_k, \; \gamma_kC_{\rm m}^k=\gamma_0C_{\rm m}^0\wh{C}_{\rm m}^k,
\end{equation}
where $\gamma_0$ and $C_{\rm m}^0$ are the representative magnitudes 
of the inverse intermembrane distance $\gamma_k$ and the capacitance 
$C_{\rm m}^k$. With this, \eqref{cmk} and \eqref{cmN} may be rewritten as:
\begin{align}\label{cmkdless}
\epsilon\wh{C}_{\rm m}^k\wh{\phi}_{kN}
&=z_0^k\wh{a}_k+\sum_{i=1}^M z_i\alpha_k\wh{c}_i^k,\; 
\wh{\phi}_{kN}=\wh{\phi}_k-\wh{\phi}_N,\\
\label{cmNdless}
-\epsilon\sum_{k=1}^{N-1}\wh{C}_{\rm m}^k\wh{\phi}_{kN}
&=z_0^N\wh{a}_N+\sum_{i=1}^M z_i\alpha_N\wh{c}_i^N,
\end{align}
where 
\begin{equation}
\epsilon=\frac{\gamma_0C_{\rm m}^0RT/F}{c_0F}.
\end{equation}
The dimensionless constant $\epsilon$ is the ratio between charge 
stored on the membrane and the bulk ionic charges. This constant 
is typically very small (on the order of $10^{-4}\sim 10^{-5}$).
To make \eqref{darcy} and \eqref{pk} dimensionless, we rescale pressure and the 
elastic force as follows.
\begin{equation}
p_k=c_0RT\wh{p}_k, \; a_k=c_0\wh{a}_k, \tau_k=\tau_0\wh{\tau}_k
\end{equation}
where $\tau_0$ is the typical magnitude of the elastic force $\tau_k$.
We may rewrite \eqref{darcy} and \eqref{pk} as:
\begin{align}\label{darcydless}
\wh{\zeta}_k\wh{\mb{u}}_k
=-\wh{\nabla}\paren{\wh{p}_k-\frac{\wh{a}_k}{\alpha_k}}
-\sum_{i=1}^N z_i\wh{c}_i^k\wh{\nabla}\wh{\phi}_k,\; 
\wh{p}_k-\wh{p}_N=A\wh{\tau}_k,
\end{align}
where 
\begin{equation}\label{Adef}
\zeta_k=\zeta_0\wh{\zeta}_k, \; A=\frac{\tau_0}{c_0RT}.
\end{equation}
The dimensionless constant $A$ is the ratio between the elastic force
and the osmotic pressure.
Finally, we may make \eqref{gateeqn} dimensionless as follows:
\begin{equation}\label{gateeqndless}
\delta\PD{s_{mg}^k}{\wh{t}}=\wh{Q}_{mg}, \; Q_{mg}=\frac{1}{\tau_g^0}\wh{Q}_{mg}, \; 
\delta=\frac{\tau_g^0}{\tau_D}.
\end{equation}
where $\tau_g^0$ is the characteristic response time of the gating variables
and $\delta$ is the ratio between the time scale of diffusion and that of 
the gating variables. This ratio is typically quite small.

\subsection{Slow Flow Limit}

Let us now discuss some limiting cases. First, consider the 
P\'eclet number ${\rm Pe}$. In the limit ${\rm Pe}\to 0$, 
all the advective terms in 
\eqref{akdless}, \eqref{aNdless}, \eqref{cikdless}, 
and \eqref{ciNdless} vanish. Furthermore, equation 
\eqref{darcydless} determining $\wh{u}_k$ is decoupled from 
the rest of the system. 
We may thus treat 
\eqref{akdless}-\eqref{ciNdless}, \eqref{cmkdless} and \eqref{cmNdless}
as equations for $\alpha_k, \wh{c}_i^k,\wh{\phi}_k$.
This is the model for which we shall develop a numerical 
scheme in Section \ref{sectnum}.
An important feature of the ${\rm Pe}\to 0$ limit is that the model 
still satisfies the energy identity \eqref{FE} with a few terms dropped.
We state this result below.
\begin{proposition}
Set ${\rm Pe}=0$ in \eqref{akdless}, \eqref{aNdless}, \eqref{cikdless}
and \eqref{ciNdless}. The variables $\alpha_k, \wh{c}_i^k,\wh{\phi}_k$ 
satisfy the dimensionless version of \eqref{FE} without 
the hydraulic dissipation term $\alpha_k\zeta_k\abs{\wh{u}_k}^2$. 
\end{proposition}
\begin{proof}
The proof is exactly the same, and simpler, than the proof of 
Theorem \ref{FEthm}.
\end{proof}

Related to the above is the limit when $A$ in \eqref{Adef} is small.
This is the limit in which the membrane is mechanically soft.
In this case, $\wh{p}_k=\wh{p}_N$ to leading order. A calculation 
analogous to the one used to derive \eqref{pNeqn} yields:
\begin{equation}
0=\wh{\nabla}\cdot
\paren{\sum_{k=1}^N\paren{\alpha_k\wh{\zeta}_k^{-1}\paren{\nabla\paren{\wh{p}_N-\frac{\wh{a}_k}{\alpha_k}}+\sum_{i=1}^Nz_i\wh{c}_i^k\wh{\nabla}\wh{\phi}_k}}}
\end{equation}
Now, suppose in addition that $\epsilon$ is small so that 
the right hand side of \eqref{cmkdless} and \eqref{cmNdless} is $0$
to leading order. Then, the above may be further rewritten as:
\begin{equation}
0=\wh{\nabla}\cdot
\paren{\sum_{k=1}^N\paren{\alpha_k\wh{\zeta}_k^{-1}\paren{\nabla\paren{\wh{p}_N-\frac{\wh{a}_k}{\alpha_k}}
-\frac{z_0^k\wh{a}_k}{\alpha_k}\wh{\nabla}\wh{\phi}_k}}}
\end{equation}
If the amount of immobile solute is low, $\wh{a}_k$ 
is small, and therefore, we find that $p_N$ satisfies a homogeneous 
elliptic equation. Given the boundary conditions \eqref{bc}, this implies 
that $p_N$ is constant everywhere. From this, it is easily seen that 
$\wh{\mb{u}_k}$ must also be $0$ to leading order. Thus, 
in the soft membrane limit, if the amount of immobile solute is low, 
we may conclude that fluid flow is negligible. 

\subsection{Electroneutral Limit and Electrotonic Effects}\label{ENEton}

The electroneutral limit is when we let $\epsilon\to 0$ in \eqref{cmkdless}
and \eqref{cmNdless}. 
These charge capacitor relations reduce to the 
electroneutrality condition. Under appropriate circumstances,
this should be a reasonable approximation given the smallness of $\epsilon$.
In this case, the electrostatic potentials
$\phi_k$ are determined so that the constraint of electroneutrality is 
satisfied at each instant of time. This electroneutral model also satisfies
the free energy identity.
\begin{proposition}
Set $\epsilon=0$ in \eqref{cmkdless} and \eqref{cmNdless}, and 
let $\wh{c}_i^k, \wh{\mb{u}}_k, \wh{\phi}_k$ and $\wh{p}_k$ satisfy 
the resulting model equations. Then, the dimensionless 
version of \eqref{FE} is satisfied 
without the capacitive energy term $C_{\rm m}\phi_{kN}^2$ in $G$.
\end{proposition}
\begin{proof}
The proof is identical to that of Theorem \ref{FEthm}.
\end{proof}
It is also possible to set both $\epsilon$ and ${\rm Pe}$ to $0$, in which 
case we again obtain a model that satisfies \eqref{FE} without the 
capacitive energy and hydraulic dissipation terms. The electroneutral 
reduction is an excellent model when fast electrophysiological 
processes (such as action potential generation) does not play 
a significant role, as we shall now see.  

Another important limit is obtained by scaling time differently.
First, let us take the derivative of \eqref{cmkdless} with 
respect to $\wh{t}$.
\begin{equation}\label{cableeqn0}
\epsilon\wh{C}_{\rm m}^k\PD{\wh{\phi}_{kN}}{\wh{t}}
=\sum_{i=1}^M z_i
\paren{-{\rm Pe}\wh{\nabla}\cdot(\alpha_k\wh{\mb{u}}_k\wh{c}_i^k)
+\wh{\nabla}\cdot\paren{\wh{D}_i^k\paren{\wh{\nabla}\wh{c}_i^k+z_i\wh{c}_i^k\wh{\nabla}\wh{\phi}_k}}
-\wh{g}_i^k},
\end{equation}
where we used \eqref{cikdless}. The above equation suggests the 
following rescaling of time:
\begin{equation}\label{tauEtauD}
t=\tau_D\wh{t}=\tau_E\wh{t}_E, \; \tau_E=\epsilon \tau_D.
\end{equation}
As we shall see, $\tau_E$ is the electrotonic time scale, in which 
cable effects are dominant. With this new scaling, \eqref{cableeqn0}
becomes:
\begin{equation}\label{cableeqn}
\wh{C}_{\rm m}^k\PD{\wh{\phi}_{kN}}{\wh{t}_E}
=\sum_{i=1}^M z_i
\paren{-{\rm Pe}\wh{\nabla}\cdot(\alpha_k\wh{\mb{u}}_k\wh{c}_i^k)
+\wh{\nabla}\cdot\paren{\wh{D}_i^k\paren{\wh{\nabla}\wh{c}_i^k+z_i\wh{c}_i^k\wh{\nabla}\wh{\phi}_k}}
-\wh{g}_i^k}.
\end{equation}
Rescaling time to $\wh{t}_E$ in 
\eqref{akdless}, \eqref{aNdless}, \eqref{cikdless}
and \eqref{ciNdless}, we see that, to leading order in $\epsilon$, 
$\wh{\alpha}_k$ and $\wh{c}_i^k$ do not change in time. 
Assume that $\wh{c}_i^k$ and $\alpha_k$ are spatially uniform initially. 
Then, $\wh{c}_i^k$ and $\alpha_k$ will remain spatially uniform in 
the $\tau_E$ time scale. We may therefore treat $\alpha_k$
and $\wh{c}_i^k$ as constants in space and time. 
Assume in addition that the P\'eclet number ${\rm Pe}\to 0$.
Then, \eqref{cableeqn} reduces to:
\begin{equation}\label{bid1}
\wh{C}_{\rm m}^k\PD{\wh{\phi}_{kN}}{\wh{t}_E}
=\wh{\nabla}\cdot\paren{\sigma_k\wh{\nabla}\wh{\phi}_k}
-\wh{I}_k, \; \sigma_k=\sum_{i=1}^M z_i^2\wh{D}_i^k\wh{c}_i^k, \; \wh{I}_k=\sum_{i=1}^M z_i\wh{g}_i^k.
\end{equation}
Likewise, we may obtain the equation for compartment $k=N$:
\begin{equation}\label{bid2}
-\sum_{k=1}^{N-1}\wh{C}_{\rm m}^k\PD{\wh{\phi}_{kN}}{\wh{t}_E}
=\wh{\nabla}\cdot\paren{\sigma_N\wh{\nabla}\wh{\phi}_N}
+\sum_{k=1}^{N-1}\wh{I}_k, \; \sigma_N=\sum_{i=1}^M z_i^2\wh{D}_i^N\wh{c}_i^N.
\end{equation}
In both \eqref{bid1} and \eqref{bid2}, $\sigma_k$ may be interpreted as the extracellular 
and intracellular conductivities, and $I_k$ is the transmembrane electric current 
flowing across the $k$-th membrane. 
We must also rescale time in \eqref{gateeqndless}:
\begin{equation}\label{bid3}
\frac{\delta}{\epsilon}\PD{s_{mg}^k}{\wh{t}_E}=\wh{Q}_{mg}.
\end{equation}
The constants $\delta$ and $\epsilon$ are typically of comparable magnitude.
If we specialize equations \eqref{bid1}, \eqref{bid2} and \eqref{bid3}
to the case $N=2$, this is nothing other than 
the bidomain equations of cardiac electrophysiology.
In the electrotonic time scale $\tau_E$, electrodiffusive effects 
are thus completely captured by electrical circuit theory, which is the 
usual starting point for deriving the bidomain equations. The bidomain 
equations are a successful model in describing action potential propagation 
in cardiac tissue.

An important property of our full system of equations, therefore, is that
it contains cable theory, or electrical circuit theory, as a submodel.
Action potential propagation is a fast electrophysiological process
in contrast to the relatively slow movement of ions that accompanies 
electrolyte and cell volume homeostasis. Our model makes it possible to 
study the interplay between the fast and slow electrophysiological processes.
The model, however, is very stiff in that it contains two disparate 
time scales, whose ratio is on the order of 
$\epsilon \approx 10^{-4}\sim 10^{-5}$.

\section{Numerical Method}\label{sectnum}

In this Section, we describe a numerical method to solve the above 
system of equations. We have developed a numerical scheme that allows for 
the solution of the above system of equations in one spatial dimension 
when there is no fluid flow (P\'eclet number ${\rm Pe}=0$).
The equations we must solve are therefore \eqref{akdless}, \eqref{aNdless}, 
\eqref{cikdless}, \eqref{ciNdless}, \eqref{cmkdless}, \eqref{cmNdless}
and \eqref{gateeqndless}. Given the presence of disparate time scales 
in the model, the model is numerically stiff. This necessitates 
the use of an implicit scheme for efficient computation.
The implicit scheme proposed here designed to satisfy discrete 
ion conservation and a discrete free energy identity.  

The dimensionless system will be used to describe our numerical method.
The symbol $\wh{\cdot}$ will be removed from all variables to 
avoid cluttered notation.
Our system is described completely by $\alpha_k, c_i^k$, and the gating 
variables $s_g^m$. Note that $\phi_k$ is determined by these variables, 
and is not needed to advance to the next time step.
We use a splitting scheme for time stepping, alternating between 
the update of $\alpha_k, c_i^k$ and of $s_g^m$. For each of these substeps,
a backward Euler type time discretization is used.  

Let $L$ be the length of the domain, $\Delta x$ be the spatial grid size
and $N_x$ be the number of grids so that $N_x\Delta x=L$.
We take a finite-volume point of view. The physical variables 
at the $l$-th grid, $(l-1)\Delta x\leq x\leq l\Delta x$, should 
be thought of as the average value over this grid, or the value 
at the midpoint of the grid.
We let the time step be $\Delta t$. 
Let $\alpha_{kl}^n,c_{il}^{kn}, \phi_{kl}^n, s_{gl}^{mn}$ be the discretized values of 
$\alpha_k,c_i^k, \phi_k$ and $s_g^m$ at the $l$-th grid at time $t=n\Delta t$. 

Let $u_l^n$ be the value of a physical quantity at the $l$-th grid, $(l-1)\Delta x\leq x\leq l\Delta x$,
and time $t=n\Delta t$. 
Introduce the following operators:
\begin{equation}
\begin{split}
\mc{D}_x^+u_l^n&=\frac{u_{l+1}^n-u_l^n}{\Delta x},\;\mc{D}_x^-u_l^n=\frac{u_l^n-u_{l-1}^n}{\Delta x}, \;
\mc{A}_x^+ u_l^n=\frac{1}{2}(u_l^n+u_{l+1}^n),\\
\mc{D}_t^- u_l^n&=\frac{u_l^n-u_l^{n-1}}{\Delta t}.
\end{split}
\end{equation}

{\em Step 1}.
In the first substep we update $\alpha_{kl}^{n}, c_{il}^{kn}$ and obtain $\phi_{kl}^n$.
We discretize equations \eqref{akdless} as follows:
\begin{equation}\label{step1alphak}
\mc{D}_t^-\alpha_{kl}^n=-w_{kl}^n, \; w_{kl}^n=\sum_{m=1}^{N_{\rm c}}w_{km}((l-1/2)\Delta x,\mb{c}_l^{kn},\mb{c}_l^{Nn},\alpha_{kl}^n,\alpha_{Nl}^n)
\end{equation}
where $\mb{c}_l^{kn}=(c_{1l}^{kn},\cdots,c_{Ml}^{kn})$ and $\mb{c}_l^{Nn}=(c_{1l}^{Nn},\cdots,c_{Ml}^{Nn})$.
we have used \eqref{waterionmemflux} and an example of the constitutive relation
for $w_{km}$ was given in \eqref{linwaterflow}. 
In place of \eqref{aNdless}, we use \eqref{incomp} for $\alpha_N$:
\begin{equation}\label{step1alphaN}
\alpha_{Nl}^n=1-\sum_{k=1}^{N-1}\alpha_{kl}^n.
\end{equation}
For equations \eqref{cikdless}, we have:
\begin{equation}\label{step1cik}
\begin{split}
&\mc{D}_t^-(\alpha_{kl}^nc_{il}^{kn})=-\mc{D}_x^- f_{il}^{kn}-g_{il}^{kn},\\
&f_{il}^{kn}=
\begin{cases}
-D_i^k(\alpha_{kl}^{n-1})(\mc{A}_x^+c_{il}^{k,n-1})
\paren{\mc{D}_x^+(\ln(c_{il}^{kn})+z_i\phi_{kl}^n)} &\text{ for } 1\leq l\leq N_x-1\\ 
0 &\text{ for } l=0, N_x.
\end{cases}
\end{split}
\end{equation}
We have set the flux $f_{il}^{kn}$ to $0$ at $l=0$ and $l=N_x$ to reflect the no-flux boundary conditions
of \eqref{bc}. 
The above discretization of the flux $f_{il}^{kn}$ was chosen so that the discrete evolution satisfies 
a discrete energy inequality similar to \eqref{FE}, as we shall see below. 
One may wonder whether 
the partially explicit treatment of the flux term in \eqref{step1cik}
may result in numerical instabilities. To address this issue, 
we have also implemented a scheme in which 
the flux term is discretized as follows:
\begin{equation}\label{step1fluxalt}
f_{il}^{kn}=
\begin{cases}
-D_i^k(\alpha_{kl}^{n-1})
\paren{\mc{D}_x^+c_{il}^{kn}+z_i(\mc{A}_x^+c_{il}^{kn})(\mc{D}_x^+\phi_{kl}^n)} &\text{ for } 1\leq l\leq N_x-1\\ 
0 &\text{ for } l=0, N_x.
\end{cases}
\end{equation}
Numerical experimentation indicates that 
the use of either \eqref{step1cik} or \eqref{step1fluxalt} does not significantly 
alter the stability properties of the numerical scheme. 

We must specify $g_{il}^{kn}$. 
\begin{equation}\label{step1gik}
\begin{split}
g_{il}^{kn}&=j_{il}^{kn}+h_{il}^{k,n-1}, \\
j_{il}^{kn}&=\sum_{m=1}^{N_{\rm c}}j_{im}^k((l-1/2)\Delta x,\mb{s}_{ml}^{k,n-1},\mb{c}_l^{kn},\mb{c}_l^{Nn},\phi_{kN,l}^n), \\
h_{il}^{k,n-1}&=\sum_{m=1}^{N_{\rm p}}h_{im}^k((l-1/2)\Delta x,\mb{c}_l^{k,n-1},\mb{c}_l^{N,n-1},\phi_{kN,l}^{n-1}),
\end{split}
\end{equation}
where $\mb{c}_l^{kn}, \mb{c}_l^{Nn},
\mb{s}_{ml}^{kn}$ and $\phi_{kN,l}^n$ are the vector of ionic concentrations in compartments $k$ and $N$, 
gating variables and the membrane potential 
evaluated at grid $l$ and time $n\Delta t$.
In the above, we used \eqref{waterionmemflux} and \eqref{hik}, and typical constitutive relations for 
$j_{im}^k$ are given in \eqref{jimk}, \eqref{HH} and \eqref{GHK}. 
Note that 
we only treat the passive flux $j_{il}^{kn}$ implicitly (but not with respect to the 
gating variables $\mb{s}$), and treat the active flux explicitly.
An implicit treatment of $j_{il}^{kn}$ is necessitated by the dissipative character of $j_{il}^{kn}$;
an explicit treatment is prone to numerical instabilities.
Equation \eqref{ciNdless} is discretized in the same way as \eqref{cikdless}:
\begin{equation}\label{step1ciN}
\mc{D}_t^-(\alpha_{Nl}^nc_{il}^{Nn})=-\mc{D}_x^- f_{il}^{Nn}+\sum_{k=1}^{N-1}g_{il}^{Nn}
\end{equation}
where $f_{il}^{Nn}$ is discretized exactly as in \eqref{cikdless}.

The capacitance-charge relation \eqref{cmkdless} 
and \eqref{cmNdless} are discretized as follows.
\begin{equation}\label{step1cm}
\begin{split}
\epsilon C_{\rm m}^k\phi_{kN,l}^{n}&=\rho_0^k+\sum_{i=1}^M z_i\alpha_{kl}^{n}c_{il}^{kn},\\
-\epsilon \sum_{k=1}^{N-1}C_{\rm m}^k\phi_{kN,l}^{n}&=\rho_0^N+\sum_{i=1}^M z_i\alpha_{Nl}^{n}c_{il}^{Nn}.
\end{split}
\end{equation}
The electrostatic potential is determined only up to a constant. This arbitrariness 
is eliminated by setting $\phi_{NN_x}^n=0$. 

The reader will realize that the scheme just described is essentially a backward Euler scheme.
We note that an explicit discretization will lead to unacceptably severe time step restrictions, 
not so much because of ionic diffusion, but because of the electrotonic diffusion of the 
membrane potential. As we discussed in Section \ref{ENEton}, our system has, embedded within it, 
the cable model or bidomain model of membrane potential propagation. The time scale for 
the spread of membrane potential is faster by a factor of $1/\epsilon$, the ratio 
between the time scales $\tau_D$ and $\tau_E$ in \eqref{tauEtauD}. The rapid electrotonic 
spread of membrane potential necessitates implicit time stepping.

The algebraic system of equations 
for the first substep thus consists of equations \eqref{step1alphak}, \eqref{step1alphaN}, 
\eqref{step1cik}, \eqref{step1gik}, \eqref{step1ciN} and \eqref{step1cm}.
We first use \eqref{step1alphaN} to eliminate $\alpha_{Nl}^n$ from the equations 
and solve the resulting algebraic system.
These equations are nonlinear, and are
solved using Newton's method. With the appropriate ordering 
of the variables, each Newton iteration results 
in a Jacobian matrix that is banded. 
The linear system is solved using a direct solver.

{\em Step 2}. In the second substep, the gating variables are updated. We discretize \eqref{gateeqndless}
as follows:
\begin{equation}\label{gateupdate}
\delta \mc{D}_t^- s_{mg,l}^{kn}=Q_{mg}(\mb{s}_{ml}^{kn},\mb{c}_l^{kn},\mb{c}_l^{Nn},\phi_{kN,l}^n).
\end{equation}
Notice that the above equation is implicit only in the gating variables $\mb{s}$ since 
the ionic concentrations and the membrane potential are known quantities as a result 
of solving the equations from Step 1.
In equation \eqref{gateupdate}, the equations for 
each grid point are decoupled, and we have only to solve a small algebraic 
system at each grid point. In the models we have implemented, 
the functions $Q_{mg}$ are 
linear in $\mb{s}$ (see \eqref{gatealphabeta} of \ref{appSDmemflux})
and it is thus a simple matter to 
solve \eqref{gateupdate}.

These two steps constitute one time step.

We note two important properties of the system of equations. 
First, we have discrete conservation of ions, in the following sense:
\begin{equation}
\mc{D}_t^-\paren{\sum_{k=1}^{N}\sum_{l=1}^{N_x} c_{il}^{kn}\Delta x}=0
\end{equation}
for all $i=1,\cdots,M$. One simple consequence of this property 
is that we also have discrete conservation of charge. Discrete conservation 
of charge is crucial for a stable numerical scheme, especially when $\epsilon$
is taken very small in \eqref{step1cm}. 
Second, we have the following discrete free energy inequality.
\begin{proposition}\label{FEdiscthm}
The solutions to \eqref{step1alphak}, \eqref{step1alphaN}, 
\eqref{step1cik}, \eqref{step1gik}, \eqref{step1ciN} and \eqref{step1cm} satisfy 
the following discrete free energy inequality.
\begin{equation}\label{FEdisc}
\begin{split}
\mc{D}_t^-G^n&\leq -I^n_{\rm bulk}-I_{\rm mem}^n, \\
G^n&=\sum_{l=1}^{N_x}\paren{\sum_{k=1}^N \paren{a_{kl}\ln\paren{\frac{a_{kl}}{\alpha_{kl}^n}}+\sum_{i=1}^M 
\alpha_{kl}^nc_{il}^{kn}\ln c_{il}^{kn}}+\sum_{k=1}^{N-1}\frac{\epsilon}{2}C_m^k(\phi_{kN,l}^n)^2}\Delta x,\\
I^n_{\rm bulk}&=\sum_{l=1}^{N_x-1}\paren{\sum_{k=1}^N\sum_{i=1}^MD_i^k(\alpha_{kl}^{n-1})\paren{\mc{A}_x^+c_{il}^{k,n-1}}
(\mc{D}_x^+\mu_{il}^{kn})^2}\Delta x,\\
I^n_{\rm mem}&=\sum_{l=1}^{N_x}\paren{\sum_{k=1}^{N-1}\paren{\psi_{kN,l}^nw_{kl}^n+\sum_{i=1}^M\mu_{il}^{kN,n}g_{il}^{kn}}}\Delta x,
\end{split}
\end{equation}
where $a_{kl}=a_k(x=(l-1/2)\Delta x)$ is the value of $a_k$ at the $l$-th grid point and 
\begin{equation}
\begin{split}
\mu_{il}^{kn}&=\ln c_{il}^{kn}+1+z_i\phi_{kl}^n, \; \mu_{il}^{kN,n}=\mu_{il}^{kn}-\mu_{il}^{Nn}\\
\psi_{kl}^{n}&=-\paren{\frac{a_{kl}}{\alpha_{kl}^n}+\sum_{i=1}^M c_{il}^{kn}}, \; \psi_{kN,l}^n=\psi_{kl}^n-\psi_{Nl}^n.
\end{split}
\end{equation}
\end{proposition}
Inequality \eqref{FEdisc} is similar to the continuous version, \eqref{FE} of Theorem \ref{FEthm}.
The crucial difference, however, is that we have a free energy {\em inequality} rather than a free energy {\em equality}.
The difficulty in the discrete case is that certain relations that are true for derivatives 
fail to hold for difference operators. With backward Euler type discretizations, however, the 
equalities fail with a definite sign so that we may still obtain inequalities.
\begin{proof}[Proof of Proposition \ref{FEdiscthm}]
The proof is essentially the same as Thoerem \ref{FEthm} except that there are certain steps 
in which equalities are replaced by inequalities. 
Multiply \eqref{step1alphak} by $\psi_{kl}^n$. The left hand side yields:
\begin{equation}
\psi_{kl}^n\mc{D}_t^-\alpha_{kl}^n=-\paren{\sum_{i=1}^M c_{il}^{kn}}\mc{D}_t^-\alpha_{kl}^n
-\frac{a_{kl}}{\alpha_{kl}^n}\mc{D}_t^-\alpha_{kl}^n.
\end{equation}
Now, 
\begin{equation}
\begin{split}
-\frac{a_{kl}}{\alpha_{kl}^n}\mc{D}_t^-\alpha_{kl}^n
=a_{kl}\paren{\frac{\alpha_{kl}^{n-1}}{\alpha_{kl}^n}-1}\geq a_{kl}\ln \paren{\frac{\alpha_{kl}^{n-1}}{\alpha_{kl}^n}}
=D_t^-\paren{a_{kl}\ln \paren{\frac{a_{kl}}{\alpha_{kl}^n}}},
\end{split}
\end{equation}
where we used the inequality:
\begin{equation}\label{lnlinineq}
\ln u\leq u-1 \text{ for } u>0.
\end{equation}
We thus have:
\begin{equation}
D_t^-\paren{a_{kl}\ln \paren{\frac{a_{kl}}{\alpha_{kl}^n}}}-\paren{\sum_{i=1}^M c_{il}^{kn}}\mc{D}_t^-\alpha_{kl}^n
\leq -\psi_{kl}^nw_{kl}^n.
\end{equation}
A similar calculation can be performed for $\alpha_{Nl}^n$. We obtain:
\begin{equation}
D_t^-\paren{a_{Nl}\ln \paren{\frac{a_{Nl}}{\alpha_{Nl}^n}}}-\paren{\sum_{i=1}^M c_{il}^{Nn}}\mc{D}_t^-\alpha_{Nl}^n
\leq \psi_{Nl}^n\sum_{k=1}^{N-1}w_{kl}^n.
\end{equation}
From the two relations above, we obtain
\begin{equation}\label{discwater}
D_t^-\paren{\sum_{k=1}^Na_{kl}\ln \paren{\frac{a_{kl}}{\alpha_{kl}^n}}}-
\sum_{k=1}^N\paren{\paren{\sum_{i=1}^M c_{il}^{kn}}\mc{D}_t^-\alpha_{kl}^n}
\leq -\sum_{k=1}^{N-1}\psi_{kN,l}^nw_{kl}^n.
\end{equation}

Let us now turn to \eqref{step1cik}. Multiply the right hand side of \eqref{step1cik} by $\mu_{il}^{kn}$.
\begin{equation}\label{discmucik1}
\mu_{il}^{kn}\mc{D}_t^-(\alpha_{kl}^nc_{il}^{kn})=(\ln c_{il}^{kn}+1)\mc{D}_t^-(\alpha_{kl}^nc_{il}^{kn})+z_i\phi_{kl}^n\mc{D}_t^-(\alpha_{kl}^nc_{il}^{kn}).
\end{equation}
Let us look at the first term.
\begin{equation}\label{discmucik2}
\begin{split}
&(\ln c_{il}^{kn}+1)\mc{D}_t^-(\alpha_{kl}^nc_{il}^{kn})\\
=&\mc{D}_t^-(\alpha_{kl}^nc_{il}^{kn}\ln c_{il}^{kn})-\alpha_{kl}^{n-1}c_{il}^{k,n-1}\ln\paren{\frac{c_{il}^{kn}}{c_{il}^{k,n-1}}}
+\alpha_{kl}^nc_{il}^{kn}-\alpha_{kl}^{n-1}c_{il}^{k,n-1}\\
\leq &
\mc{D}_t^-(\alpha_{kl}^nc_{il}^{kn}\ln c_{il}^{kn})-\alpha_{kl}^{n-1}c_{il}^{k,n-1}\paren{\frac{c_{il}^{kn}}{c_{il}^{k,n-1}}-1}
+\alpha_{kl}^nc_{il}^{kn}-\alpha_{kl}^{n-1}c_{il}^{k,n-1}\\
=&
\mc{D}_t^-(\alpha_{kl}^nc_{il}^{kn}\ln c_{il}^{kn})+c_{il}^{kn}D_t^-\alpha_{kl}^n,
\end{split}
\end{equation}
where we used \eqref{lnlinineq} in the above inequality.
Sum the second term on the right hand side of \eqref{discmucik1} in $i$.
\begin{equation}\label{discmucik3}
\sum_{i=1}^Mz_i\phi_{kl}^n\mc{D}_t^-(\alpha_{kl}^nc_{il}^{kn})=\epsilon C_{\rm m}^k\phi_{kl}^n\mc{D}_t^-\phi_{kN,l}^n.
\end{equation}
Combining \eqref{discmucik1}, \eqref{discmucik2} and \eqref{discmucik3} we obtain:
\begin{equation}
\begin{split}
&\sum_{i=1}^M\paren{\mc{D}_t^-(\alpha_{kl}^nc_{il}^{kn}\ln c_{il}^{kn})+c_{il}^{kn}D_t^-\alpha_{kl}^n}+\epsilon C_{\rm m}^k\phi_{kl}^n\mc{D}_t^-\phi_{kN,l}^n\\
&\leq \sum_{i=1}^M \mu_{il}^{kn}\mc{D}_t^-(\alpha_{kl}^nc_{il}^{kn})
=-\sum_{i=1}^M\mu_{il}^{kn}\mc{D}_x^-f_{il}^{kn}-\sum_{i=1}^M\mu_{il}^{kn}g_{il}^{kn}.
\end{split}
\end{equation}
The last equality follows from \eqref{step1cik}.
We can obtain a similar inequality for $k=N$ using \eqref{step1ciN}, 
and combine this with the above inequality. This yields: 
\begin{equation}\label{discmucik}
\begin{split}
&\sum_{k=1}^N\sum_{i=1}^M\paren{\mc{D}_t^-(\alpha_{kl}^nc_{il}^{kn}\ln c_{il}^{kn})+c_{il}^{kn}D_t^-\alpha_{kl}^n}+\sum_{k=1}^{N-1}\epsilon C_{\rm m}^k\phi_{kN,l}^n\mc{D}_t^-\phi_{kN,l}^n\\
\leq &-\sum_{k=1}^N\sum_{i=1}^M\mu_{il}^{kn}\mc{D}_x^-f_{il}^{kn}-\sum_{k=1}^{N-1}\sum_{i=1}^M\mu_{il}^{kN,n}g_{il}^{kn}.
\end{split}
\end{equation}
It is easily seen that the second sum of the first line satisfies the inequality:
\begin{equation}\label{disccapE}
\sum_{k=1}^{N-1}\epsilon C_{\rm m}^k\phi_{kN,l}^n\mc{D}_t^-\phi_{kN,l}^n\geq \sum_{k=1}^{N-1}\frac{\epsilon}{2}C_{\rm m}^k\mc{D}_t^-(\phi_{kN,l}^n)^2.
\end{equation}
Combining \eqref{discwater}, \eqref{discmucik} and \eqref{disccapE}, we obtain:
\begin{equation}\label{FEtemp1}
\begin{split}
&\mc{D}_t^-\paren{\sum_{k=1}^N \paren{a_{kl}\ln\paren{\frac{a_{kl}}{\alpha_{kl}^n}}+\sum_{i=1}^M 
\alpha_{kl}^nc_{il}^{kn}\ln c_{il}^{kn}}+\sum_{k=1}^{N-1}\frac{\epsilon}{2}C_m^k(\phi_{kN,l}^n)^2}\\
\leq& -\sum_{k=1}^{N-1}\paren{\psi_{kN,l}^nw_{kl}^n+\sum_{i=1}^M\mu_{il}^{kN,n}g_{il}^{kn}}
-\sum_{k=1}^N\sum_{i=1}^M\mu_{il}^{kn}\mc{D}_x^-f_{il}^{kn}
\end{split}
\end{equation}
Note that
\begin{equation}\label{FEtemp2}
\sum_{l=1}^{Nx}\sum_{k=1}^N\sum_{i=1}^M\mu_{il}^{kn}\mc{D}_x^-f_{il}^{kn} \Delta x
=-\sum_{l=1}^{Nx-1}\sum_{k=1}^N\sum_{i=1}^M(\mc{D}_x^+\mu_{il}^{kn})f_{il}^{kn}\Delta x=I_{\rm bulk},
\end{equation}
where we summed by parts in the first equality and used the expression for $f_{il}^{kn}$ in 
\eqref{step1cik} in the second equality. We obtain the desired 
inequality by multiplying \eqref{FEtemp1} by $\Delta x$ and summing in $l$, 
and combining this with \eqref{FEtemp2}.
\end{proof}
Inequality \eqref{FEdisc} ensures that the discrete free energy increases can only come from the 
active flux contribution $h_{il}^{kn}$. Indeed, $I_{\rm bulk}$ is non-negative and the contributions 
from $w_{kl}^n$ and $j_{il}^{kn}$ in $I_{\rm mem}$ are also non-negative given the implicit treatment 
of $w_{kl}^n$ and $j_{il}^{kn}$ (see \eqref{step1alphak} and \eqref{step1gik}) and the 
structural conditions for $w_k$ and $j_i^k$ (see \eqref{dissip}).

\section{Simulation of Cortical Spreading Depression}\label{sectSD}

\subsection{Model Setup}\label{sectSDmodel}
We apply the above model to a computation of cortical spreading depression.
The equations, specialized to this application, will be relisted here 
(in dimensional form) to facilitate discussion.
We treat neural tissue as a biphasic continuum following 
\cite{tuckwell1978mathematical,yao2011continuum}, 
so that we have two compartments ($N=2$). 
Compartment $1$ or ${\rm i}$ is the intracellular (neuronal) 
and compartment $2$ or ${\rm e}$ is the extracellular compartment
(we shall thus use $1,2$ and ${\rm i,e}$ interchangeably for 
subscripts/superscripts of our variables). 
We neglect 
fluid flow, and equations \eqref{alphak} and \eqref{alphaN} are thus
\begin{equation}\label{alphasdeq}
\PD{\alpha_{\rm i}}{t}=-\PD{\alpha_{\rm e}}{t}=-\gamma w.
\end{equation}
Here and in the following, we omit the compartmental subscripts associated with 
membrane quantities (we have only two compartments, and thus only one 
membrane, the neuronal membrane).
The transmembrane water flux $w$ will be specified shortly.

We consider three ionic species Na$^+$, K$^+$ and Cl$^-$.
Equations for ionic concentrations \eqref{cik}, \eqref{ciN} and \eqref{fik}
reduce to
\begin{align}
\PD{(\alpha_{\rm i}c_i^{\rm i})}{t}&=
\PD{}{x}\paren{D_i^{\rm i}\paren{\PD{c_i^{\rm i}}{x}+\frac{z_iFc_i^{\rm i}}{RT}\PD{\phi_{\rm i}}{x}}}-\gamma g_i
\label{ciSD}\\
\PD{(\alpha_2c_i^{\rm e})}{t}&=
\PD{}{x}\paren{D_i^{\rm e}\paren{\PD{c_i^{\rm e}}{x}+\frac{z_iFc_i^{\rm e}}{RT}\PD{\phi_{\rm e}}{x}}}+\gamma g_i
\label{ceSD}
\end{align}
where $i=1,2,3$ corresponding to Na$^+$, K$^+$ and Cl$^-$ respectively.
Following \cite{yao2011continuum}, we let the diffusion coefficient 
in the extracellular space be given by:
\begin{equation}
D_i^{\rm e}=D_i^*\alpha_{\rm e}
\end{equation}
where $D_i^*$ is the diffusion coefficient in aquaeous solution. The diffusion 
coefficient in the extracellular space thus decreases with volume fraction. 
The diffusion coefficient in the intracellular space 
$D_i^{\rm i}$ reflects gap junction connectivity. We let
\begin{equation}\label{chiDi}
D_i^{\rm i}=\chi D_i^*
\end{equation}
where $\chi$ is a constant to be varied in the simulations to follow.
The electrostatic potentials $\phi_1$ and $\phi_2$ are specified by the following capacitance 
charge relation \eqref{cmk} and \eqref{cmN}
\begin{equation}\label{chargecapSD}
\gamma C_{\rm m}\phi_{\rm m}=z_0^{\rm i}Fa_{\rm i}+\sum_{i=1}^3 z_iF\alpha_{\rm i}c_i^{\rm i}
=-\paren{z_0^{\rm e}Fa_{\rm e}+\sum_{i=1}^3 z_iF\alpha_{\rm e}c_i^{\rm e}}, \; 
\phi_{\rm m}=\phi_{\rm i}-\phi_{\rm e}.
\end{equation}
Constants that appear in the above equations are listed 
in Table \ref{model_params}. The amount of impermeable ions, $a_{\rm i}$
and $a_{\rm e}$ are specified together with the initial data 
(see \eqref{imp_solutes} of \ref{appSDinit}.)

\begin{table}
\begin{center}
\begin{tabular}{|c|c|c|c|}
\hline
$\gamma^{-1}$(cm) &$1.5662\times 10^{-4}$ & $D_{\rm Na}^*$ (cm$^2$/s)&$1.33\times 10^{-5}$\\
$\wh{\eta}_{\rm w}$(cm$/$s/(mmol$/$l))& $5.4\times 10^{-2}$& $D_{\rm K}^*$ (cm$^2$/s)&$1.96\times 10^{-5}$\\
$C_{\rm m}$ ($\mu$F/cm$^2$) & 0.75 & $D_{\rm Cl}^*$ (cm$^2$/s)&$2.03\times 10^{-5}$\\
$T$ (K$^\circ$)& 310.15 & $z_0$ & -1\\
\hline
\end{tabular}
\caption{Model parameters. Standard values are used 
for the Faraday constant $F$ and ideal gas constant $R$.
\label{model_params}}
\end{center}
\end{table}

Transmembrane water flow $w$ in \eqref{alphasdeq} is given by
the constitutive relation (see \eqref{linwaterflow})
\begin{equation}\label{wsd}
w=\eta^{\rm w}\paren{\pi_{{\rm we}}-\pi_{{\rm wi}}}=\wh{\eta^{\rm w}}
\paren{\frac{a_{\rm e}}{\alpha_{\rm e}}+\sum_{i=1}^3 c_i^{\rm e}-\frac{a_{\rm i}}{\alpha_{\rm i}}-\sum_{i=1}^3 c_i^{\rm i}}.
\end{equation}
We have set the elastic force to be $\tau_{\rm i}=\tau_1=0$ so that 
$\psi_{\rm i,e}=\psi_{12}=\psi_1-\psi_2$ is equal to $\pi_{{\rm w}2}-\pi_{{\rm w}1}$ (see \eqref{pk},
\eqref{psik}). The value of $\wh{\eta}_{\rm w}$ 
is given in Table \ref{model_params}.
Prescription \eqref{wsd} 
is essentially equivalent to that in 
\cite{yao2011continuum,shapiro2001osmotic}, except that 
we do not impose the constraint that $\alpha_{\rm i}$ must not exceed $0.95$.
As $\alpha_1$ approaches $1$, $\alpha_{\rm e}=1-\alpha_{\rm i}$ approaches $0$ and 
thus $\pi_{{\rm we}}$ grows large so long as $a_{\rm e}>0$. 
The resulting large osmotic force does not allow $\alpha_{\rm i}$ to 
become arbitrarily close to $1$. 

We use the ion channel models of 
\cite{kager2000simulated,kager2002conditions,yao2011continuum} 
for our simulations
which we describe in \ref{appSDmemflux}. 
Specification of initial data is discussed in \ref{appSDinit}.

\subsection{Simulation Results}\label{sectsimresults}

We set the length $L$ of our one-dimensional domain to be equal to $1$cm.
To initiate a spreading depression wave, excitatory fluxes $j_{i{\rm E}}$
are added as in \eqref{transmembranefluxes}. We set
\begin{equation}\label{excitation}
\begin{split}
j_{iE}&=G_{\rm E}(t,x)(\mu_i^{\rm i}-\mu_i^{\rm e}),\\
G_{\rm E}(t,x)&=
\begin{cases}
G_{\rm max}\cos^2(\pi x/2L_{\rm E})\sin(\pi t/t_{\rm E}) &\text{ if } 
0\leq t<t_{\rm E} \text{ and } 0\leq x<L_{\rm E},\\
0 & \text{ otherwise}.
\end{cases}
\end{split}
\end{equation}
We set $L_{\rm E}=0.1$cm, $t_{\rm E}=2$s and $G_{\rm max}F^2=0.5$mS$/$cm$^2$.
Thus a non-selective membrane conductance opens up for a 
brief period at the left edge of the domain.

In the numerical simulations to follow,
the number of spatial grid points is taken to be $N_x=500$
and $\Delta t=10$ms. 

A sample computation is shown in Figure \ref{csd_profile1}, 
where there is no gap junctional connectivity ($\chi=0$ in \eqref{chiDi}).
A wave of SD depolarization, accompanied by 
a large increase in K$^+$ concentration, is initiated near $x=0$
and propagates to the positive $x$ direction. 
We point out that our SD computation produces 
a negative shift in the extracellular voltage (known as the 
negative DC shift). This is, to the best of our knowledge, 
the first time this quantity has been 
computed in a biophysically consistent fashion
(there are some previous attempts in computing the negative DC 
shift in the literature \cite{almeida2004modeling,bennett2008quantitative}; 
the relationship between this 
and our present approach is discussed in Appendix \ref{appExtVol}) 
This is significant given the importance of the negative 
DC shift as an experimental signal in the detection of SD.
We computed the speed of the SD wave as follows.
At each grid point, we may compute the time at which the membrane 
potential reaches a threshold value of $-30$mV. We then use these 
values at grid points that fall in the interval $L/5<x<L/2$ to 
compute the speed of the wave. For the computations shown in 
Figure \ref{csd_profile1}, 
the wave speed is $5.56$cm/min, which is within the range 
of physiologically plausible values.

\begin{figure}
\begin{center}
\includegraphics[width=\textwidth]{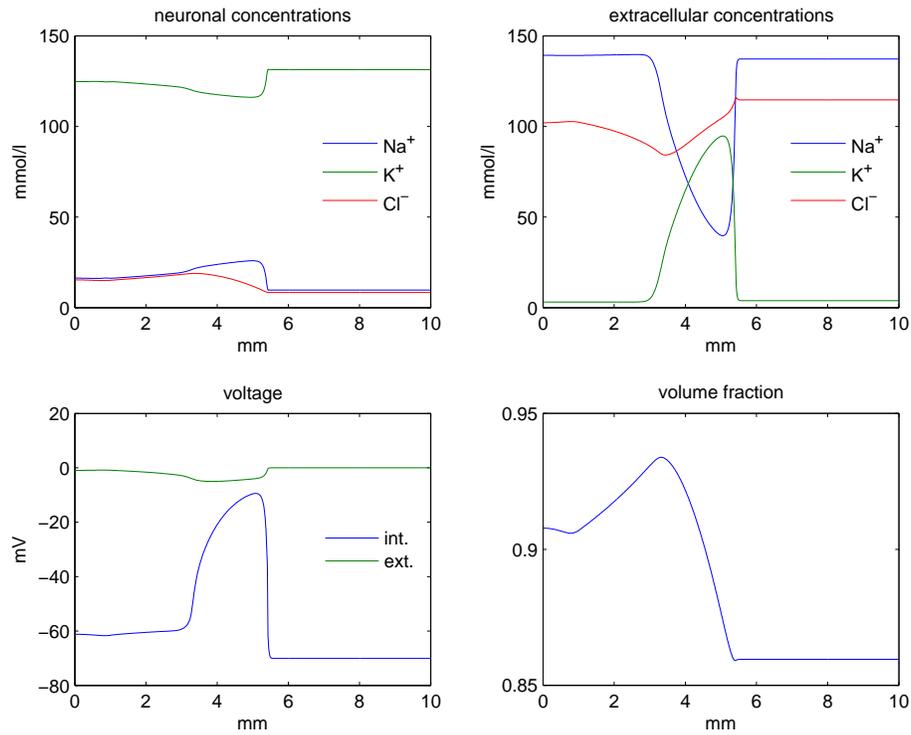}
\caption{A snapshot of an SD wave at $t=50$s.
Plotted are intracellular and extracellular ionic concentrations, 
intracellular(int.) and extracellular(ext.) voltages and the intracellular 
volume fraction. Note that the extracellular voltage experiences 
a negative shift (the negative DC shift).
\label{csd_profile1}}
\end{center}
\end{figure}

\subsection{Varying gap junctional conductance}

We study the dependence of the SD wave speed 
on the strength of gap junctional conductance. 
It has been suggested that gap junctional conductance 
may be necessary for the propagation of SD waves \cite{somjen2004ions}, 
and this was tested using a computational model in \cite{shapiro2001osmotic}.
Here, we reexamine this hypothesis.

We vary the value of $\chi$ in \eqref{chiDi}
from $0$ to $10^{-3}$ in increments of $5\times 10^{-5}$. 
Note that, in \cite{shapiro2001osmotic}, $\chi$ was given a value of $1/4$.
The resulting SD wave speed is given in Figure \ref{gap_speed_fig}.
We see that even a small increase in gap junctional conductance 
(far smaller than that postulated in \cite{shapiro2001osmotic}), 
leads to propagation 
speeds that exceed physiologically realistic bounds by large margins
(typical speeds are $2$ to $7$cm/min). 
The likely reason for the discrepancy between our computations
and those of \cite{shapiro2001osmotic} is that electrotonic coupling is not 
properly accounted for in \cite{shapiro2001osmotic}. Gap junctional coupling 
will inevitably lead to cable (or electrotonic) effects, 
which will enable fast wave propagation as seen in cardiac 
or skeletal muscle tissue. Constitutively open gap junctions, 
therefore, are likely not involved in the propagation of 
SD waves. For the gap junctional hypothesis to be viable, 
closed gap junctions may have to open with the spread of the 
wave \cite{somjen2004ions}.

\begin{figure}
\begin{center}
\includegraphics{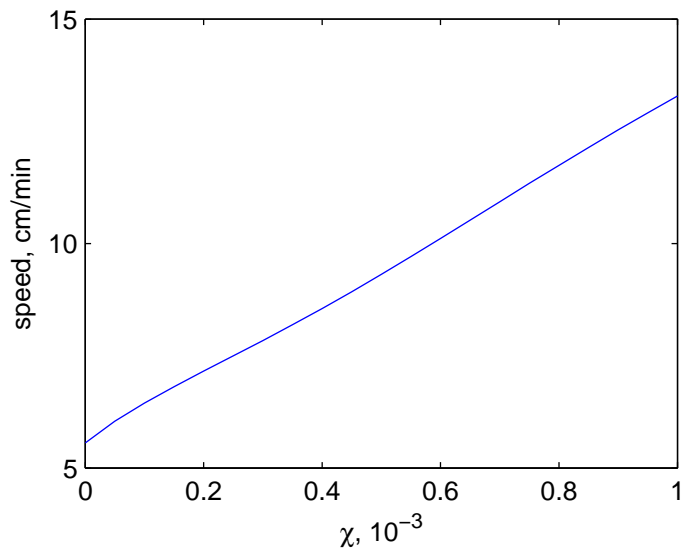}
\caption{Speed of spreading depression wave as a function of the 
parameter $\chi$ in \eqref{chiDi}.
\label{gap_speed_fig}}
\end{center}
\end{figure}

\subsection{Varying extracellular chloride concentration}

The value of the extracellular chloride concentration can be variable,
and its effect on SD is not well-understood. 
Here, we vary the preparatory initial value of extracellular
chloride concentration $c_{{\rm Cl}*}^{\rm e}$ between 
$6$mmol/l and $120$mmol/l and
perform computations at $31$
logarithmically equi-spaced values.

A sample plot of the propagating front when 
$c_{{\rm Cl}*}^{\rm e}=6$mmol/l is given in Figure \ref{csd_profile2}.
There are several interesting differences between this and the case 
$c_{{\rm Cl}*}^{\rm e}=120$mmol/l (shown in Figure \ref{csd_profile1}).
First, the spreading depression wave form is altered. 
The wave in the $c_{{\rm Cl}*}^{\rm e}=6$mmol/l case has longer wavelength,
and thus, a longer duration at each spatial location.
Another difference is that in the $c_{{\rm Cl}*}^{\rm e}=6$mmol/l case, 
the change in neuronal volume is small. 
Given (near) electroneutrality, osmotic pressure change is 
possible only when both anions and cations can pass the membrane.
With little chloride, inward Na$^+$ flux cannot be accompanied 
by a matching inward Cl$^-$ flux. This is in line with 
the verbal arguments in \cite{somjen2004ions}.

\begin{figure}
\begin{center}
\includegraphics[width=\textwidth]{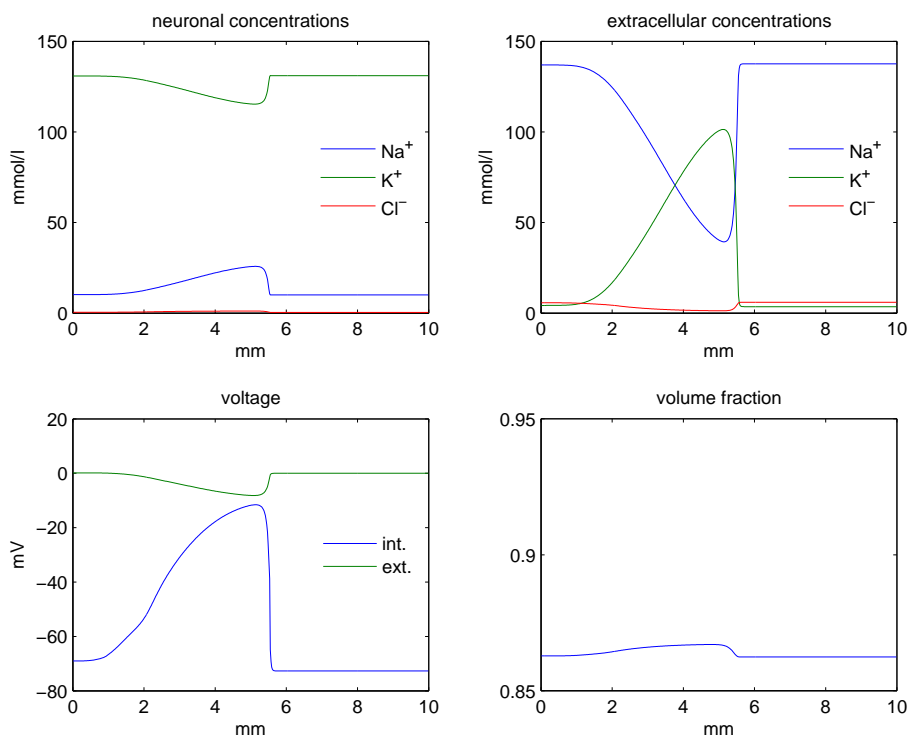}
\caption{A snapshot of an SD wave at $t=50$s when 
$c_{{\rm Cl}*}^{\rm e}=6$mmol/l. Compared to Figure \ref{csd_profile1}, 
the wave is wider and the volume change is minimal.
\label{csd_profile2}}
\end{center}
\end{figure}

In Figure \ref{Cl_speed_fig}, we plot the SD propagation 
speed as a function of $c_{{\rm Cl},0}^{\rm e}$. It is interesting 
that the dependence is non-monotonic. The reason why the speed 
increases at low $c_{{\rm Cl},0}^{\rm e}$ is likely because a high 
chloride concentration has a stabilizing effect on membrane excitability. 
The reason for the increase in speed at higher chloride concentration 
may be due to the fact that higher extracellular 
chloride concentration facilitates potassium diffusion.
In order for potassium to diffuse, by (near) electroneutrality, 
chloride must also diffuse, or a deficit in sodium concentration 
must be created. The speed of these processes should influence 
the ease with which potassium can diffuse, and thus, 
the speed of the SD wave.

\begin{figure}
\begin{center}
\includegraphics{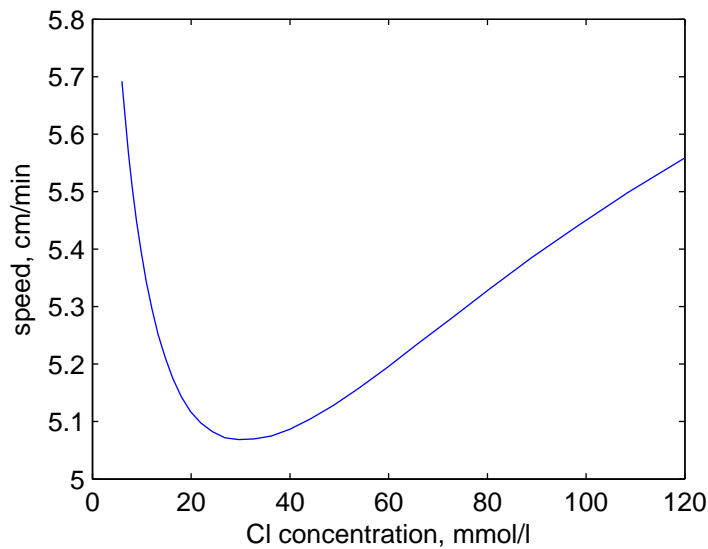}
\caption{Speed of spreading depression wave as a function of 
$c_{{\rm Cl}*}^{\rm e}$.
\label{Cl_speed_fig}}
\end{center}
\end{figure}

\section{Conclusion}

In this paper, we formulated a multidomain tissue model of ionic electrodiffusion, 
volume changes and osmotic water flow. We devised a numerical scheme for one spatial dimension 
without interstitial flow. This was applied to the study of SD.

An interesting theoretical issue is the relation of this tissue level 
model to more microscopic cellular level models such as \cite{mori_liu_eis}. 
The cardiac bidomain model can be derived as a formal homogenization limit of 
a microscopic model \cite{neu1993hst,KS}, 
and a similar derivation may be possible here.  

There is much to be done in terms of numerical algorithms. 
In the brain, it is increasingly recognized that water 
flow may play an important physiological role \cite{nedergaard2013garbage}, 
and it is thus of great interest 
to develop a numerical scheme that can treat water flow. 
The algorithm presented in Section \ref{sectnum} easily generalizes to 
two and three spatial dimensions, but the required computational cost may 
be substantial and much work may be needed for the development of efficient solvers.  
Another important direction would be to devise numerical methods that exploit 
the presence of disparate time scales, by updating certain variables at finer 
time steps than others.

An important feature of the model was that it satisfies an energy identity, 
and this property may be of direct interest in the study of SD. Indeed,
SD is understood as a major breakdown in ionic homeostasis, 
or dissipation of actively stored free energy \cite{dreier2013spreading}. 
Our model 
provides a means of quantitatively computing this breakdown.

The SD model used here is limited in several respects, the most 
important of which is the absence of a glial compartment,
which is known to play a significant role in ionic concentration 
homeostasis and hence in SD \cite{somjen2004ions}.

Finally, it should be stressed that the multidomain electrodiffusion model 
formulated here is not restricted in its application to SD or to brain 
ionic homeostasis. We hope it would find application in many physiological 
systems both neural and beyond.\newline

\noindent\textbf{Acknowledgments} I would like to thank 
Bob Eisenberg and Chun Liu for valuable discussion. 
Bob Eisenberg directed 
the author's attention to interstitial flow. 
Huaxiong Huang and Wei Yao kindly provided their code used 
in their publication \cite{yao2011continuum}.
I also thank 
the Fields Institute (Toronto, Canada) for the generous support during 
the Spreading Depression workshop in the summer of 2014. 
Many participants of the workshop gave me valuable input 
and much encouragement. This paper is dedicated to Robert Miura,
who introduced me to this topic almost ten years ago.  

\appendix

\section{Details of Spreading Depression Simulation}\label{appSD}

\subsection{Transmembrane Fluxes}\label{appSDmemflux}
We follow \cite{kager2000simulated,kager2002conditions,yao2011continuum}  for the transmembrane fluxes. We have:
\begin{equation}
\begin{split}\label{transmembranefluxes}
g_{\rm Na}&=j_{\rm NaL}+j_{\rm NaP}+j_{\rm NaE}+2h_{\rm NaK},\\
g_{\rm K}&=j_{\rm KL}+j_{\rm KDR}+j_{\rm KA}+j_{\rm KE}-3h_{\rm NaK},\\
g_{\rm Cl}&=j_{\rm ClL}+j_{\rm ClE}.
\end{split}
\end{equation}
The leak flux $j_{i{\rm L}}$ have the following form (see \eqref{HH}):
\begin{equation}
j_{i{\rm L}}=G_i(\mu_i^{\rm i}-\mu_i^{\rm e})
\end{equation}
where the conductances $G_i(z_iF)^2$ are given in Table \ref{leak_NaKATPase}.
\begin{table}
\begin{center}
\begin{tabular}{|c|c||c|c|}
\hline
ion & conductance (mS$/$cm$^2$)& \multicolumn{2}{|c|}{NaK ATPase parameters}\\
\hline
Na$^+$& $2\times 10^{-2}$&$I_{\rm max}$ & $13$ ($\mu$A/cm$^2$)\\
K$^+$& $7\times 10^{-2}$&$K_{\rm K}$ & $2$ (mmol/l)\\
Cl$^-$& $2\times 10^{-1}$&$K_{\rm Na}$ & $7.7$ (mmol/l)\\
\hline 
\end{tabular}
\caption{Leak conductances and NaKATPase parameters\label{leak_NaKATPase}}
\end{center}
\end{table}
The persistent Na$^+$ flux $j_{\rm NaP}$ has the following form
(see \eqref{GHK})
\begin{equation}
j_{\rm NaP}=m_{\rm NaP}^2h_{\rm NaP}P_{\rm NaP}
J_{\rm GHK}(1,c_{\rm Na}^{\rm i},c_{\rm Na}^{\rm e},\phi_{\rm m})
\end{equation}
where $P_{\rm NaP}$ is the permeability and 
$s=m_{{\rm NaP}}, h_{{\rm NaP}}$ are the gating variables.
The gating variables satisfy the equations:
\begin{equation}
\PD{s}{t}=\alpha_s(\phi_m)(1-s)-\beta_s(\phi_m)s.\label{gatealphabeta}
\end{equation}
The form of $j_{\rm KA}$ and $j_{\rm KDR}$ are similar.
The parameters and functions defining the above equations are
given in Table \ref{gated_fluxes}.

\begin{table}
\begin{center}
\begin{tabular}{|c|c|c|rcl|}
\hline
flux & $P$ (cm/s) & gates & \multicolumn{3}{|c|}{rate functions (ms$^{-1}$)}\\
\hline
\multirow{4}{*}{$j_{\rm NaP}$} & \multirow{4}{*}{$2\times 10^{-5}$} 
& \multirow{4}{*}{$m^2h$} 
& $\alpha_m$&$=$&$(6(1+\exp(-(0.143\phi_m+5.67))))^{-1}$ \\
& & & $\beta_m$&$=$&$1-\alpha_m$\\
& & & $\alpha_h$&$=$&$5.12\times 10^{-6}\exp(-(0.056\phi_m+2.94))$\\
& & & $\beta_h$&$=$&$1.6\times10^{-4}(1+\exp(-(0.2\phi_m+8)))^{-1}$\\
\hline
\multirow{2}{*}{$j_{\rm KDR}$} & \multirow{2}{*}{$1\times 10^{-3}$} 
& \multirow{2}{*}{$m^2$} & $\alpha_m$&$=$&$0.08\varphi(0.2\phi_{\rm m}+6.98)$ \\
& & & $\beta_m$&$=$&$0.25\exp(-(0.25\phi_{\rm m}+1.25))$\\
\hline
\multirow{4}{*}{$j_{\rm KA}$} & \multirow{4}{*}{$1\times 10^{-4}$} 
& \multirow{4}{*}{$m^2h$} & $\alpha_m$&$=$&$0.2\varphi(0.1\phi_{\rm m}+5.69)$ \\
& & & $\beta_m$&$=$&$0.175\wh{\varphi}(0.1\phi_{\rm m}+2.99)$\\
& & & $\alpha_h$&$=$&$0.016\exp(-(0.056\phi_{\rm m}+4.61))$\\
& & & $\beta_h$&$=$&$0.5(1+\exp(-(0.2\phi_{\rm m}+11.98)))^{-1}$\\
\hline
\end{tabular}
\caption{Ion fluxes and their corresponding parameters and rate functions.
In the above, $P$ is the permeability, 
$\varphi(u)=u/(1-\exp(-u)), \wh{\varphi}(u)=u/(\exp(u)-1)$ and 
the membrane potential $\phi_{\rm m}$ is in mV.\label{gated_fluxes}}
\end{center}
\end{table}

The excitation currents $j_{i{\rm E}}$ are used to initiate the spreading
depression wave. This is described in \eqref{excitation} of Section 
\ref{sectsimresults}. 

The Na$^+$ and K$^+$ flux carried by the NaK ATPase is given by $3h_{\rm NaK}$
and $-2h_{\rm NaK}$ respectively in \eqref{transmembranefluxes}.
Here, $h_{\rm NaK}$ is given by
\begin{equation}
h_{\rm NaK}=\wh{I}_{\rm max}(1+K_{\rm K}/c_{\rm K}^{\rm e})^{-2}(1+K_{\rm Na}/c_{\rm Na}^{\rm i})^{-3}
\end{equation}
where the constants $I_{\rm max}=\wh{I}_{\rm max}F, K_{\rm K}$ and $K_{\rm Na}$ 
are given in Table \ref{leak_NaKATPase}.

\subsection{Initial Conditions}\label{appSDinit}

We first set preparatory initial data and run the model 
to steady state. These steady state values are then used 
as initial data to run the model simulations (with $0$ excitatory fluxes).

The list of preparatory initial data for the concentrations $c_i^k$
and membrane potential $\phi_{\rm m}$, and volume fraction $\alpha_k$
are given in Table \ref{conc_al_ivals}. 
The preparatory initial value for intracellular chloride is given by 
the expression
\begin{equation}\label{pinitclint}
c_{{\rm Cl}*}^{\rm i}=c_{{\rm Cl}*}^{\rm e}\exp(\phi_{{\rm m}*}F/RT)
\end{equation}
where the subscript $*$ refers to the preparatory initial values.
Once these preparatory initial value 
are given, we may compute the impermeable solute amount $a_k$ by 
solving \eqref{EN} for $a_k$:
\begin{equation}\label{imp_solutes}
a_k=-\frac{1}{z_0^k}\sum_{i=1}^M z_i\alpha_k^*c_{i*}^k.
\end{equation}

\begin{table}
\begin{center}
\begin{tabular}{|c|c|c|c|}
\hline
$\alpha_{\rm i}$ & $1/1.15$ & $\phi_{\rm m}$ & $-70$(mV)\\
\hline
$c_{\rm Na}^{\rm i}$ & $10$ & $c_{\rm Na}^{\rm e}$ & $145$ \\
$c_{\rm K}^{\rm i}$ & $130$ & $c_{\rm K}^{\rm e}$ & $3.5$\\
$c_{\rm Cl}^{\rm i}$ & $-$ & $c_{\rm Cl}^{\rm e}$ & $120$\\
\hline
\end{tabular}
\caption{Preparatory initial values. Concentrations are 
in mmol/l. For intracellular chloride concentration, 
see \eqref{pinitclint}. Note that $\alpha_{\rm e}$ is 
set to $1-\alpha_{\rm i}$ (see \eqref{incomp}).\label{conc_al_ivals}}
\end{center}
\end{table}

The preparatory initial values of the gating variables are set 
to the steady state values of \eqref{gatealphabeta}:
\begin{equation}
s=\frac{\alpha_s(\phi_{{\rm m}*})}{\alpha_s(\phi_{{\rm m}*})+\beta_s(\phi_{{\rm m}*})}.
\end{equation}

Given these preparatory initial conditions, the model is run 
to steady state with no excitatory fluxes ($j_{i{\rm E}}=0$ in 
\eqref{transmembranefluxes}) and $\Delta t=10$s. 
The preparatory run is terminated 
when the discrete time derivative of the ionic concentrations falls 
below $10^{-12}$ times the maximum ionic concentration. We note that 
the difference between the preparatory initial values and the steady 
state values are typically very small. 

\section{Computation of Extracellular Voltage}\label{appExtVol}

In our model, the extracellular voltage is computed as a natural 
output of the system of equations, and we cannot, in general, compute the 
membrane potential without computing both the extracellular 
and intracellular voltages (and the other compartmental voltages
if there are more than two compartments). 
There is, however, a special situation 
in which the membrane potential can be computed without computing 
the extracellular voltage. We discuss this special case, as it 
relates to previous attempts in obtaining the extracellular 
voltage \cite{almeida2004modeling,bennett2008quantitative}.
Let us restrict our attention to the two compartment case 
without fluid flow in one spatial dimension. We let the equations be satisfied 
on the interval $0<x<L$.
We adopt the notation 
of Section \ref{sectSDmodel}. Let us assume furthermore 
that gap junctional coupling is absent ($D_i^{\rm i}=0$).
Taking the time derivative of the first equality in \eqref{chargecapSD}
and using \eqref{ceSD}, we have:
\begin{equation}\label{pointmodelreduction}
\gamma C_{\rm m}\PD{\phi_{\rm m}}{t}=\sum_{i=1}^M \gamma g_i,
\end{equation}
where we used our assumption $D_i^{\rm i}=0$.
The above equation does not explicitly depend on $\phi_{\rm i}$ or $\phi_{\rm e}$, 
and only on the membrane potential $\phi_{\rm m}$, since the transmembrane 
fluxes $g_i$ depend on voltage only through $\phi_{\rm m}$.
Now, let us use the electroneutrality relation in place of the charge capacitor 
relation \eqref{chargecapSD}:
\begin{equation}\label{ENSD}
0=z_0^{\rm i}Fa_{\rm i}+\sum_{i=1}^3 z_iF\alpha_{\rm i}c_i^{\rm i}
=-\paren{z_0^{\rm e}Fa_{\rm e}+\sum_{i=1}^3 z_iF\alpha_{\rm e}c_i^{\rm e}}, \; 
\phi_{\rm m}=\phi_{\rm i}-\phi_{\rm e}.
\end{equation}
Then \eqref{pointmodelreduction} reduces further to:
\begin{equation}\label{ENediffform}
\sum_{i=1}^M \gamma g_i=0.
\end{equation}
Equations \eqref{pointmodelreduction} and \eqref{ENediffform} are often used 
in modeling studies to obtain the membrane potential. Note, however, that 
this is valid only when there is no gap junctional coupling.

Let us now take the derivative of the second equality in \eqref{ENSD} with respect to $t$.
Using \eqref{ciSD} and \eqref{ENediffform}, we have
\begin{equation}
\PD{}{x}\paren{a+\sigma\PD{\phi_{\rm e}}{x}}=0, \; a=\sum_{i=1}^M z_iFD_i^{\rm e}\PD{c_i^{\rm e}}{x}, \; 
\sigma=\sum_{i=1}^M \frac{(z_iF)^2D_i^{\rm e}c_i^{\rm e}}{RT}.
\end{equation}
This is the same as \eqref{cableeqn} except that the capacitor term and the advective current terms
are absent. Assuming no-flux boundary conditions at $x=0$ and $x=L$, we obtain, from the above:
\begin{equation}
a+\sigma\PD{\phi_{\rm e}}{x}=0.
\end{equation}
This is the relation used to determine the extracellular voltage
in \cite{almeida2004modeling,bennett2008quantitative}.
It should be emphasized, however, that one may use the 
above expression to compute
the extracellular voltage only under the restrictive 
conditions of no gap junctional coupling, 
one-dimensional geometry and no-flux boundary 
conditions. Otherwise, the charge capacitor relation (or 
equivalently, near electroneutrality) will be violated.

\bibliographystyle{plain}
\bibliography{mylib}
\end{document}